\documentclass[journal,comsoc]{IEEEtran}

\usepackage{graphicx,epsfig}
\usepackage[noadjust]{cite}
\usepackage{mcite}
\usepackage{amsfonts,helvet}
\usepackage{fancyhdr}
\usepackage{threeparttable}
\usepackage{epsf}
\usepackage{amsthm}
\usepackage{amsmath}
\usepackage{siunitx}
\usepackage{amssymb}
\usepackage{booktabs}

\usepackage{dsfont}
\usepackage{subfigure}
\usepackage{color}
\usepackage[linesnumbered,ruled, noend]{algorithm2e}
\usepackage{algpseudocode}
\usepackage{algcompatible}
\usepackage{enumerate}
\usepackage{ragged2e}
\usepackage{algorithm2e}
\usepackage{optidef}

\usepackage{wrapfig}
\usepackage{cancel}

\newtheorem{theorem}{Theorem}

\newtheorem{lemma}{Lemma}

\newtheorem{proposition}{Proposition}
\newtheorem{remark}{Remark}

\newtheorem{assumption}{Assumption}

\usepackage{eucal}

\usepackage{graphicx,epsfig}
\usepackage[noadjust]{cite}
\usepackage{mcite}
\usepackage{amsfonts,helvet}
\usepackage{fancyhdr}
\usepackage{threeparttable}
\usepackage{epsf}
\usepackage{amsthm}
\usepackage{amsmath}
\usepackage{siunitx}
\usepackage{amssymb}
\usepackage{booktabs}
\usepackage{dsfont}
\usepackage{subfigure}
\usepackage{color}
\usepackage[linesnumbered,ruled, noend]{algorithm2e}
\usepackage{algpseudocode}
\usepackage{algcompatible}
\usepackage{enumerate}
\usepackage{cancel}
\usepackage{multirow} 
\usepackage{threeparttable}

\usepackage{graphicx,epsfig}
\usepackage[noadjust]{cite}
\usepackage{mcite}
\usepackage{amsfonts,helvet}
\usepackage{fancyhdr}
\usepackage{threeparttable}
\usepackage{epsf,epsfig}
\usepackage{amsthm}
\usepackage{amsmath}
\usepackage{siunitx}
\usepackage{amssymb}
\usepackage{booktabs}
\usepackage{dsfont}
\usepackage{subfigure}
\usepackage{color}
\usepackage[linesnumbered,ruled,noend]{algorithm2e}
\usepackage{algpseudocode}
\usepackage{algcompatible}
\usepackage{enumerate}
\usepackage{cancel}
\usepackage{bbm}
\usepackage{graphicx,subfigure}
\usepackage{graphicx}
\usepackage{array}
\newcolumntype{P}[1]{>{\centering\arraybackslash}p{#1}}
\usepackage{eucal}

\def\ba{{\bf a}}

\def\be{{\bf e}}
\def\bff{{\bf f}}
\def\bg{{\bf g}}
\def\bh{{\bf h}}

\def\bq{{\bf q}}

\def\bs{{\bf s}}

\def\bu{{\bf u}}
\def\bv{{\bf v}}
\def\bw{{\bf w}}
\def\bx{{\bf x}}
\def\by{{\bf y}}
\def\bz{{\bf z}}

\def\bA{{\bf A}}
\def\bB{{\bf B}}
\def\bC{{\bf C}}

\def\bF{{\bf F}}
\def\bG{{\bf G}}
\def\bH{{\bf H}}
\def\bI{{\bf I}}

\def\bK{{\bf K}}

\def\bN{{\bf N}}

\def\bP{{\bf P}}

\def\bS{{\bf S}}

\def\bU{{\bf U}}
\def\bV{{\bf V}}
\def\bW{{\bf W}}

\def\cC{\mbox{$\mathcal{C}$}}

\def\cN{\mbox{$\mathcal{N}$}}

\def\bbC{\mbox{$\mathbb{C}$}}

\def\bbE{\mbox{$\mathbb{E}$}}

\begin{document}
\title{Scalable and Convergent Generalized Power Iteration Precoding for Massive MIMO Systems}

\author{Seunghyeong Yoo, Mintaek Oh, Jeonghun Park, Namyoon Lee, and Jinseok Choi
\thanks{S. Yoo, M. Oh, and J. Choi are with the Department of Electrical Engineering, Korea Advanced National Institute of Science and Technology (KAIST), Daejeon 34141, South Korea (e-mail: {\texttt{\{seunghyeong, ohmin, jinseok\}@kaist.ac.kr}}). 
J. Park is with the School of Electrical and Electronic Engineering, Yonsei University, Seoul 03722, South Korea (e-mail: {\texttt{jhpark@yonsei.ac.kr}}).
N. Lee is with the Department of Electrical Engineering, Pohang University of Science and Technology (POSTECH), Pohang 37673, South Korea (e-mail: {\texttt{nylee@postech.ac.kr}}).
    }
}

\maketitle
\setcounter{page}{1} 
\begin{abstract}
In massive multiple-input multiple-output (MIMO) systems, achieving high spectral efficiency (SE) often requires advanced precoding algorithms whose complexity scales rapidly with the number of antennas, limiting practical deployment.
In this paper, we develop a scalable and computationally efficient generalized power iteration precoding (GPIP) framework for massive MIMO systems under both perfect and imperfect channel state information at the transmitter (CSIT).
By exploiting the low-dimensional subspace property of optimal precoders, we reformulate the high-dimensional beamforming problem into a lower-dimensional weight optimization that scales with the number of users rather than antennas.
We further extend this framework to the imperfect CSIT scenario by showing that stationary solutions reside in a combined subspace spanned by the estimated channel and error covariance matrices, enabling a robust design via low-rank approximation.
To reduce computational cost, we leverage the Sherman–Morrison formula to simplify matrix inversions.
Moreover, interpreting the GPIP update as a projected preconditioned gradient ascent method, we establish convergence guarantees and develop a stable and monotonic algorithm using a backtracking line search.
Numerical results demonstrate that the proposed methods achieve the highest SE performance compared to state-of-the-art linear precoders with significantly reduced complexity and convergence, highlighting their suitability for large-scale MIMO systems.

\begin{IEEEkeywords}
    Low-complexity precoding, massive MIMO, and imperfect CSIT, and convergence analysis.
\end{IEEEkeywords}
\end{abstract}

\section{Introduction}
Massive multiple-input multiple-output (MIMO) is a key technology for next-generation networks.
By equipping a base station (BS) with a large antenna array, it spatially multiplexes numerous users on the same time-frequency resource, offering substantial improvements in spectral efficiency (SE) and energy efficiency (EE) \cite{marzetta2016book:massive, bjornson2017TWC:massive, zhang2020JSAC:massive}.
However, these benefits entail a significant computational cost.
As the number of antennas scales up, the computational complexity required for advanced beamforming grows rapidly, creating a critical bottleneck for practical deployment in cellular systems.
Consequently, the prohibitive processing burden remains a common and severe hurdle for cellular deployments, necessitating a unified and scalable solution.
\subsection{Related Works}
Beamformer design for MIMO systems has evolved from complex nonlinear methods to efficient linear and advanced optimization-based techniques.
Early research focused on nonlinear approaches like zero-forcing dirty paper coding (ZF-DPC) \cite{caire2003TIT:DPC}, which theoretically achieves Shannon capacity but entails prohibitive computational complexity for massive multi-user MIMO (MU-MIMO) systems.
This practical limitation necessitated low-complexity linear techniques, such as regularized ZF (RZF) \cite{peel2005TCOM:RZF}.
While computationally efficient, linear methods are inherently suboptimal for sum SE maximization due to the non-convex nature of the signal-to-interference-plus-noise ratio (SINR).

To address these limitations, various iterative optimization algorithms have been developed.
The weighted minimum mean square error (WMMSE) method \cite{christensen2008TWC:WMMSE} pioneered this domain by transforming the weighted sum SE maximization problem into a tractable quadratically constrained quadratic program via the minimization of the mean square error.
Building on this approach, \cite{yalcin2018TWC:WMMSE:analysis} proposed low-complexity algorithms for MU-MIMO systems that simultaneously transmit common and unicast data.
In \cite{shen2018TSP:FP}, optimization algorithms based on quadratic transform (QT) principles tackled the non-convex problem by reformulating it into a sequence of convex subproblems.
Similarly, the minorize-maximize (MM) framework \cite{lu2019TCOM:MM:massive} was leveraged to construct surrogate functions for precoder design.
In \cite{choi2019TWC:GPI}, the generalized power iteration precoding (GPIP) algorithm was proposed by deriving the first-order Karush--Kuhn--Tucker (KKT) optimality conditions for the sum SE maximization problem and reformulating them as generalized eigenvalue problems.

Recent advancements have further refined these frameworks.
The GPIP algorithms have demonstrated superior SE and EE across diverse scenarios, including quantized massive MIMO \cite{choi2022TWC:GPI:massive}, secure transmission with artificial noise \cite{choi2022TCOM:GPI:AN}, rate-splitting multiple access (RSMA) \cite{park2023TWC:GPI:RSMA}, the finite blocklength regime \cite{oh2023TWC:GPI:FBL}, and integrated sensing and communications (ISAC) with imperfect channel state information (CSI) \cite{choi2024TWC:GPI:ISAC}.
In the realm of fractional programming (FP), research has focused on robust precoding for imperfect CSI \cite{lin2022commlett:FP:massive} and new FP algorithms \cite{shen2024JSAC:fastFP:general} that are computationally more efficient than the conventional quadratic transform by eliminating large matrix inversions.
These new approaches were further extended to ISAC systems in \cite{chen2025TWC:fastFP:ISAC}.
A comprehensive review \cite{shen2025arXiv:FP} also detailed the theory and applications of QT, including these recently developed techniques \cite{shen2024JSAC:fastFP:general, chen2025TWC:fastFP:ISAC}.
Meanwhile, the MM framework was leveraged to design linear precoders for downlink (DL) massive MIMO systems \cite{zhang2022TVT:MM:massive}, robust precoding for reconfigurable intelligent surface (RIS)-aided systems experiencing channel aging \cite{he2025TVT:MM:RIS}, and distributed methods for network massive MIMO that restrict information exchange \cite{zhu2023WCNC:MM:massive}. 
Parallel efforts in WMMSE-based algorithms addressed robust precoding under limited feedback \cite{zhou2024TWC:WMMSE:robust} and channel uncertainty for massive MIMO \cite{shi2023TCOM:WMMSE:massive}, as well as for RIS-assisted systems \cite{chen2022TVT:WMMSE:RIS, choi2024TCOM:WMMSE:RIS}.

Despite these significant advances, a fundamental scalability bottleneck remains in the existing literature. Most optimization-based precoding frameworks \cite{christensen2008TWC:WMMSE, yalcin2018TWC:WMMSE:analysis, shen2018TSP:FP, lu2019TCOM:MM:massive, choi2019TWC:GPI} incur computational complexity that scales cubically with the number of BS antennas, rendering them impractical for large-scale massive MIMO deployments. 
While more recent QT- and WMMSE-based approaches have reduced the complexity to quadratic \cite{shen2024JSAC:fastFP:general} or even linear scaling \cite{chen2025TWC:fastFP:ISAC, zhao2023TSP:R-WMMSE}, such improvements critically rely on the restrictive assumption of perfect CSI at the transmitter (CSIT).

Notably, the reduced-WMMSE (R-WMMSE) framework \cite{zhao2023TSP:R-WMMSE} demonstrated that exploiting the intrinsic low-dimensional subspace structure of the optimal precoder can substantially improve computational scalability. However, despite the strong empirical performance and broad applicability of the GPIP method compared to several alternative approaches \cite{park2022rate, park2023TWC:GPI:RSMA}, a comparable structural advancement has not yet been achieved for GPIP. In its conventional form, GPIP still incurs computational complexity that scales cubically with the number of BS antennas and lacks rigorous convergence guarantees. As a result, although GPIP often provides superior performance and greater generality, it remains computationally prohibitive and theoretically underdeveloped for large-scale massive MIMO systems.

These observations motivate us to revisit the GPIP framework from both scalability and theoretical standpoints with extended applicability to the imperfect CSIT case. 
By bridging this gap, our work provides a unified precoding framework that preserves the performance advantages of GPIP, achieves practical scalability in massive MIMO systems, and ensures provable algorithmic convergence.

\subsection{Contributions}
In this work, we aim to develop a scalable GPIP formulation whose complexity scales with the number of users rather than the number of antennas under both perfect and imperfect CSIT. 
Moreover, we rigorously establish, for the first time, the convergence of the GPIP algorithm to stationary points, further proposing a scalable and convergent GPIP framework.
The main contributions are summarized as follows:
\begin{itemize}
    \item Under perfect CSIT, we develop a scalable GPIP (S-GPIP) framework by exploiting the low-dimensional subspace structure of optimal precoders. By reformulating the original high-dimensional beamforming problem into a reduced-dimensional weight optimization, the proposed method achieves computational complexity that is scalable with respect to the number of antennas.
    \item We extend the proposed framework to the imperfect CSIT scenario by proving that any nontrivial stationary solution resides in the combined subspace spanned by the estimated channel and error covariance matrices. Leveraging a low-rank approximation of the error covariance, we develop a robust and scalable S-GPIP algorithm with controlled complexity growth.
    \item We further reduce the per-iteration computational burden by exploiting the block-diagonal structure of the KKT matrices and applying the Sherman–Morrison formula. This reduces the algorithm complexity from $\CMcal{O}(K^4)$ to $\CMcal{O}(K^3)$ regarding the number of users $K$, enabling practical implementation in large-scale massive MIMO systems.
    \item We provide the first rigorous convergence analysis of the GPIP framework by interpreting the update rule as a projected preconditioned gradient ascent (PPGA) method. Under mild assumptions, we establish convergence to stationary points and propose a convergent S-GPIP algorithm to guarantee stable and monotonic ascent. 
\end{itemize}

In Section~\ref{sec:system}, we introduce the system model and formulate the sum SE maximization problem. 
In Section~\ref{sec:S-GPI}, we develop an S-GPIP algorithm under perfect CSIT.
In Section~\ref{sec:S-GPI-imperfect}, we extend the framework to imperfect CSIT. 
In Section~\ref{sec:conv_pg}, we provide convergence analysis and propose a convergent GPIP algorithm. 
Section~\ref{sec:Numerical_results} presents numerical results validating the performance, scalability, and convergence of the proposed algorithms.
Finally, Section~\ref{sec:conclusion} concludes the paper.

$\mathit{Notations}$: $a$ is a scalar, $\ba$ is a vector, and $\bA$ is a matrix. $\bI_{N}$ is an identity matrix of size $N \times N$, {$\mathbf{0}_{N}$ is a zero vector of size $N \times 1$}, $\bA \otimes \bB$ is the Kronecker product of two matrices $\bA$ and $\bB$, $\be_k$ is the $k$th basis column vector, whose $k$th entry is one and all other entries are zero, and $\cC\cN(\mu,\sigma^{2})$ is a complex Gaussian distribution with mean $\mu$ and variance $\sigma^{2}$. Superscripts $(\cdot)^{\sf T}$, $(\cdot)^{\sf H}$, and $(\cdot)^{-1}$ denote a matrix transpose, Hermitian, matrix inverse, respectively. The operation $\bbE[\cdot]$, $\mathrm{tr}(\cdot)$, $\mathrm{rank}(\bA)$, and $\mathrm{diag}({\bf a})$ are expectation, trace of a matrix, rank of a matrix, and a diagonal matrix with elements of $\bf a$.

\section{System Model and Problem Formulation}
\label{sec:system}
We consider a single-cell DL massive MIMO system, where the BS equipped with $N$ transmit antennas serves $K$ DL users ($N\gg K$), each having a single receive antenna.
Let $\bs \! = \! [s_{1}, ..., s_{K}]^{\sf T} \! \in \! \bbC^{K}$ denote a DL transmit symbol vector, where  $s_k$ is a symbol for user $k$ and each symbol follows an independent and identically distributed (IID) complex Gaussian distribution, i.e., $s_{k}\sim \cC\cN(0,1)$.

The transmit symbol vector is precoded by a digital precoder $\bF =[\bff_{1},...,\bff_{K}] \! \in \! \bbC^{N \times K}$, where $\bff_{k} \! \in \! \bbC^{N}$ is the $k$th column of $\bF$.
Then, the precoded symbol vector is written as
\begin{align}
    \label{eq:symbol_vector}
    \bx = \sqrt{P} \sum_{k=1}^{K} \bff_k s_k = \sqrt{P} \bF \bs,
\end{align}
where $P$ is the maximum transmit power at the BS.
Accordingly, a received signal at user $k$ is given by
\begin{align}
    \label{eq:received_sig}
    y_{k} = \sqrt{P} \bh^{\sf H}_{k} \bff_{k} s_{k} + \sqrt{P} \sum_{i=1,i\ne k}^{K}\bh^{\sf H}_{k} \bff_{i} s_{i} + n_{k}, 
\end{align}
where $\bh_{k} \! \in \! \bbC^{N}$ is the $k$th column of a channel matrix $\bH \! = \! [\bh_{1},..., \bh_{K}] \! \in \! \bbC^{N \times K}$ and $n_{k}$ is the additive white Gaussian noise (AWGN) for user $k$, which follows an IID complex Gaussian distribution $\cC\cN(0, \sigma^{2})$.
Furthermore, we model the channel vector for user $k$ as $\bh_{k}$ which has zero mean and covariance matrix of $\bK_{\bh_{k}}$.

Accordingly, we define the sum SE as a performance metric and formulate the corresponding optimization problem, where the closed-form SE expression for the $k$th user is defined as
\begin{align}
    \label{eq:rate}
    R_{k}(\bF) = \log_{2} \left( 1 + \frac{|\bh^{\sf H}_{k} \bff_{k}|^{2}}{\sum_{i=1, i\ne k}^{K} |\bh^{\sf H}_{k} \bff_{i}|^{2} + \frac{\sigma^{2}}{P}} \right).
\end{align}
From \eqref{eq:rate}, we formulate the following optimization problem:
\begin{subequations}
    \label{problem}
    \begin{align}
        \underset{\bF}{\text{maximize}}& \;\; \sum_{k=1}^{K} R_{k}(\bF) \label{eq:problem}\\
        \text{subject to}& \;\; \mathrm{tr} \left( \bF \bF^{\sf H} \right) \leq 1, \label{eq:constraint}
    \end{align}
\end{subequations}
where the inequality constraint in \eqref{eq:constraint} denotes the maximum transmit power budget of the BS.

Maximizing the sum SE in \eqref{problem} is a non-convex problem where finding the global optimum is intractable. 
While linear precoders (e.g., maximum-ratio transmission (MRT) and ZF) offer low computational complexity, they are highly suboptimal, necessitating advanced iterative algorithms.
However, in massive MIMO systems where $N \gg K$, direct optimization of the precoder $\bF$ faces severe complexity bottlenecks due to high dimensionality: {\it{(i)}} The complexity of state-of-the-art iterative algorithms (e.g., WMMSE, GPIP) typically scales cubically with $N$, i.e., $\CMcal{O}(N^3)$ \cite{christensen2008TWC:WMMSE, choi2019TWC:GPI}, making them prohibitive for practical implementation.
{\it{(ii)}} The substantial processing requirements translate to excessive energy consumption and impose strict hardware constraints on large-scale arrays.

Instead of optimizing the high-dimensional precoder $\bF$ directly, we propose to reformulate the problem within a significantly lower-dimensional subspace.
We exploit the critical insight that optimal precoding vectors must lie within the subspace spanned by the columns of the channel matrix $\bH$ \cite{zhao2023TSP:R-WMMSE}.
Based on this projection, we propose the S-GPIP method that is scalable to the number of BS antennas for both perfect and imperfect CSIT scenarios.
To guarantee convergence, we further perform convergence analysis and propose a convergent S-GPIP algorithm.

\section{Proposed Scalable Precoding}
\label{sec:S-GPI}

\subsection{Low-dimensional Projection}
We first design the precoding algorithm for the perfect CSIT scenario.
Our design is based on the following proposition from \cite{zhao2023TSP:R-WMMSE}, which characterizes the nontrivial stationary points of the problem in \eqref{problem}:
\begin{proposition}[\cite{zhao2023TSP:R-WMMSE}]
    \label{proposition1}
    Any nontrivial stationary point $\bff_{k}$ of the problem in \eqref{problem} must lie in the range space of the DL channel $\bH$, i.e., $\bff_{k} = \bH \bw_{k}$, with a unique vector $\bw_{k} \in \bbC^{K}$ for $k \in \{1,...,K\}$.
\end{proposition}

From Proposition~\ref{proposition1}, we observe that each digital precoding vector $\bff_{k}$ for the $k$th user can be expressed as a weighted linear combination of channel vectors.
Motivated by this finding, we design a digital precoder as a function of the channel, i.e., 
\begin{align}
    \label{eq:opt_precoder}
    \bF = \bH \bW,
\end{align}
where $\bW = [\bw_{1}, ..., \bw_{K}] \in \bbC^{K \times K}$ is a unique weight matrix for the channel matrix and each column $\bw_{k}$ corresponds to the $k$th user.
Although the optimization problem in \eqref{problem} imposes an inequality power constraint, the sum SE generally increases with the total transmit power.
Thus, we assume that the BS operates at the maximum transmit power, implying that the power constraint is satisfied with equality, i.e., $\text{tr}(\bF\bF^{\sf H}) = 1$. 
To explicitly reflect this assumption and facilitate the subsequent problem reformulation, we normalize the precoder in \eqref{eq:opt_precoder} as
\begin{align}
    \label{eq:normalization}
    \frac{\bF}{\|\bF\|_{\sf F}} = \frac{\bH \bW}{\sqrt{\mathrm{tr} \left(\bH \bW \bW^{\sf H} \bH^{\sf H} \right)}}.
\end{align}
Based on Proposition~\ref{proposition1} and \eqref{eq:normalization}, we reformulate the SINR of the $k$th user in \eqref{eq:rate}:
\begin{align}
    \label{eq:re_SINR}
    \gamma^{\sf Ref}_{k} = \frac{|\bh^{\sf H}_{k} \bH \bw_{k}|^{2}}{\sum_{i=1, i\ne k}^{K} |\bh^{\sf H}_{k} \bH \bw_{i}|^{2} + \frac{\sigma^{2}}{P} \mathrm{tr} \left(\bH \bW \bW^{\sf H} \bH^{\sf H} \right) } .
\end{align}

Accordingly, under the assumption of $\text{tr}(\bF\bF^{\sf H}) = 1$, the optimization problem of \eqref{problem} can be reformulated as
\begin{subequations}
    \label{re_problem1}
    \begin{align}
        \underset{\bW}{\text{maximize}}&\;\; \sum_{k=1}^{K} R^{\sf Ref}_k(\bW) \label{eq:re_problem1}\\
        \text{subject to}&\;\; \mathrm{tr} \left( \bH \bW \bW^{\sf H} \bH^{\sf H} \right) = 1,
        \label{eq:re_constraint1}
    \end{align}
\end{subequations}
where $R^{\sf Ref}_k(\bW) = \log_2{ ( 1 + \gamma^{\sf Ref}_{k} ) }$.
Leveraging Proposition~\ref{proposition1} significantly reduces the feasible search domain for the optimal solution; however, the reformulated problem of \eqref{re_problem1} remains inherently non-convex.
To address this challenge, we employ a new objective function of \eqref{eq:re_problem1} along with a tailored solution strategy, which enables convergence to a superior stationary point despite the non-convex nature of the problem.

\subsection{Reformulation and Algorithm}
We first reformulate the SE of the objective function in \eqref{eq:re_problem1} to adopt the GPIP algorithm \cite{choi2019TWC:GPI}.
Let $\bN_{k}$ be
\begin{align}
    \label{eq:Nk}
    \bN_{k} = \bH^{\sf H} \bh_{k} \bh_{k}^{\sf H} \bH \in \bbC^{K \times K}.
\end{align} 
Next, we define a tall weight vector by stacking $\bw_{k}$ as
\begin{align}
    \label{eq:stacked_weight_vec}
    \bar{\bw} = \left[\bw_{1}^{\sf T}, ..., \bw_{K}^{\sf T} \right]^{\sf T} \in \bbC^{K^{2}}.
\end{align}
Based on \eqref{eq:stacked_weight_vec}, we rewrite a trace term of the normalization in \eqref{eq:re_SINR} to be a tractable form to the GPIP scheme, i.e.,
\begin{align}
    \label{eq:re_normalization}
    \mathrm{tr} \left(\bH \bW \bW^{\sf H} \bH^{\sf H} \right) = \bar{\bw}^{\sf H} \left(\bI_{K} \otimes \bH^{\sf H} \bH \right) \bar{\bw}.
\end{align}
Using \eqref{eq:Nk}, \eqref{eq:stacked_weight_vec}, and \eqref{eq:re_normalization}, we represent the SE expression of the $k$th user in \eqref{eq:re_problem1} as a Rayleigh quotient form:
\begin{align}
    \label{eq:rayleigh_form}
    R^{\sf Ref}_{k}(\bar{\bw}) = \mathrm{log}_{2} \left(\frac{\bar{\bw}^{\sf H} \bA_{k} \bar{\bw}}{\bar{\bw}^{\sf H} \bB_{k} \bar{\bw}} \right),
\end{align}
where $\bA_{k}$ and $\bB_{k}$ are obtained as
\begin{align}
    \label{eq:Ak}
    &\bA_{k} = \bI_{K} \otimes \bN_{k} + \frac{\sigma^{2}}{P} \left(\bI_{K} \otimes \bH^{\sf H} \bH \right) \in \bbC^{K^{2} \times K^{2}},\\
    \label{eq:Bk}
    &\bB_{k} = \bA_{k} - \be_k \be_k^{\sf H} \otimes \bN_{k} \in \bbC^{K^{2} \times K^{2}}.
\end{align}

From \eqref{eq:rayleigh_form}, the optimization problem in \eqref{re_problem1} is expressed as
\begin{align}
    \label{eq:re_problem}
    \underset{\bar{\bw}}{\text{maximize}}&\;\; \sum_{k=1}^{K} R^{\sf Ref}_{k}(\bar{\bw}).
\end{align} 
We remark that since the objective function \label{re_problem} is scale-invariant with respect to ${\bf \bar{w}}$, we omit the power constraint at this point.
Note that the reformulated problem \eqref{eq:re_problem} remains non-convex, making the search for a global optimum NP-hard.
Therefore, we derive the first-order optimality condition of \eqref{eq:re_problem} to identify the stationary points of the problem, as stated in the following lemma:

\begin{lemma}
    \label{Lemma1}
    The first-order optimality condition of \eqref{eq:re_problem} is represented in the generalized eigenvalue problem as
    \begin{align}
        \bB_{\sf KKT}^{-1}(\bar{\bw})\bA_{\sf KKT}(\bar{\bw}) = \lambda(\bar{\bw})\bar{\bw},
        \label{eq:NEPv}
    \end{align}
    where
    \begin{align}
        \bA_{\sf KKT}(\bar{\bw}) &= \sum_{k=1}^{K}\left(\frac{\bA_{k}}{\bar{\bw}^{\sf H} \bA_{k} \bar{\bw}} \right) \lambda_{\sf num}(\bar{\bw}),
        \label{eq:A_KKT}
        \\
        \bB_{\sf KKT}(\bar{\bw}) &= \sum_{k=1}^{K}\left(\frac{\bB_{k}}{\bar{\bw}^{\sf H} \bB_{k} \bar{\bw}} \right) \lambda_{\sf den}(\bar{\bw}),
        \label{eq:B_KKT}
        \\
        \lambda(\bar{\bw})  &= \prod_{k=1}^{K} \left(\frac{\bar{\bw}^{\sf H} \bA_{k} \bar{\bw}}{\bar{\bw}^{\sf H} \bB_{k} \bar{\bw}} \right).
        \label{eq:eig_val}
    \end{align}
    The numerator $\lambda_{\sf num}(\bar{\bw})$ and denominator $\lambda_{\sf den}(\bar{\bw})$ of eigenvalue are any functions that satisfy $\lambda(\bar{\bw})=\lambda_{\sf num}(\bar{\bw})/\lambda_{\sf den}(\bar{\bw})$.
    \begin{proof}
        See Appendix~\ref{app:Lemma1}.
    \end{proof}
\end{lemma}

\begin{algorithm}[!t]
\caption{Scalable GPIP under Perfect CSIT}\label{alg:one}
\bf{initialize}: $\bar{\bw}^{(0)} \gets$ \rm{RZF}\\
\rm{Set the iteration count} $t=1.$\\
\While{$ \|\bar{\bw}^{(t)} - \bar{\bw}^{(t-1)}\| > \varepsilon_{\bar{\bw}}$ \rm{and} $t \leq t_{\sf max}$}{
\rm{Build matrix} $\bA_{\sf KKT} (\bar{\bw}^{(t-1)})$ according to \eqref{eq:A_KKT}.\\
\rm{Build matrix} $\bB_{\sf KKT} (\bar{\bw}^{(t-1)})$ according to \eqref{eq:B_KKT}.\\
\rm{Compute} $\bar{\bw}^{(t)}$ = $\bB_{\sf KKT}^{-1}  (\bar{\bw}^{(t-1)}) \bA_{\sf KKT} (\bar{\bw}^{(t-1)}) \bar{\bw}^{(t-1)}.$\\
Normalize $\bar{\bw}^{(t)} = {\bar{\bw}^{(t)}}/{\|\bar{\bw}^{(t)}\|}.$\\
$t \gets t+1.$\\
}
$\bar{\bw}^{\star} \gets \bar{\bw}^{(t)}.$\\
\rm{Build weight matrix} $\bW^{\star} \gets \bar{\bw}^{\star}$.\\
\rm{Compute} $\bF = \bH \bW^{\star}$.\\
\rm{Normalize} $\bF^{\star} = \bF/\|\bF\|_{\sf{F}}$.\\
\Return $\bF^{\star}$.
\end{algorithm}
The proposed S-GPIP method is summarized in Algorithm~\ref{alg:one}.
The key steps are as follows:
Using a conventional RZF precoder, we first initialize the stacked weight vectors $\bar{\bw}^{(0)}$.
At each iteration $t$, we build the matrices $\bA_{\sf KKT} (\bar{\bw}^{(t-1)} )$ and $\bB_{\sf KKT} (\bar{\bw}^{(t-1)})$ according to \eqref{eq:A_KKT} and \eqref{eq:B_KKT}.
Then, we compute $\bar{\bw}^{(t)}$ as $\bar{\bw}^{(t)} \! = \! \bB_{\sf KKT}^{-1} (\bar{\bw}^{(t-1)} ) \bA_{\sf KKT} (\bar{\bw}^{(t-1)}) \bar{\bw}^{(t-1)}$ and normalize $\bar{\bw}^{(t)} \! = \! \bar{\bw}^{(t)}/ \|\bar{\bw}^{(t)}\|$.
We repeat this process until the stopping criterion $\|\bar{\bw}^{(t)} \! - \! \bar{\bw}^{(t-1)}\| \! \leq \! \varepsilon_{\bar{\bw}}$ is met or until the iteration count $t$ reaches the maximum count $t_{\sf max}$, where $\epsilon_{\bar{\bw}} \! > \! 0$ is a tolerance level.
Using $\bar{\bw}^{\star}$, we build a optimized weight matrix $\bW^{\star}$.
Finally, we compute $\bF = \bH \bW^{\star}$, and then normalize the precoder $\bF^{\star} = \bF/{\|\bF\|_{\sf F}}$.


\subsection{Complexity Reduction}
\label{subsec:Complexity_Analysis}

It is important to convert the original sum SE maximization problem, which involves optimizing the precoder $\bF$, into a reformulated problem of optimizing the unique weight matrix $\bW$.
Based on this transformation, Algorithm~\ref{alg:one} provides a low-dimensional optimization framework, which significantly reduces the computational complexity.
The computational complexity of Algorithm~\ref{alg:one} is primarily determined by the matrix inversion $\bB_{\sf KKT}^{-1}(\bar{\bw}^{(t-1)})$, which is an inherent property of GPIP methods.
As in \cite{choi2019TWC:GPI}, without the low-dimensional framework, the matrix $\bB_{\sf KKT}$ is typically an $NK \times NK$ block-diagonal matrix, resulting in an inversion complexity of $\CMcal{O}(N^{3}K)$ per-iteration.
However, under our framework, $\bB_{\sf KKT}$ is a $K^{2} \times K^{2}$ block-diagonal matrix, thereby reducing the per-iteration inversion complexity to $\CMcal{O}(K^{4})$.
The total complexity, including $T$ iterations and the $\CMcal{O}(NK^2)$ complexity for pre- and post-processing (e.g., $\bH^{\sf H} \bH$ and $\bF = \bH \bW^{\star}$), is $\CMcal{O}(TK^4 + NK^2)$ instead of $\CMcal{O}(TN^3K)$.
Considering $N\gg K$, this significant reduction in complexity makes the proposed S-GPIP highly advantageous for deployment in large-scale MIMO systems.

We can further reduce the complexity from $\CMcal{O}(TK^4 + NK^2)$ to $\CMcal{O}(TK^3 + NK^2)$.
To this end, we first demonstrate how the Sherman--Morrison formula \cite{golub2013book:Sherman} facilitates efficient matrix inversion.
Assuming $\lambda_{\sf den}(\bar{\bw}) = \prod_{i=1}^{K} ( {\bar{\bw}^{\sf H} \bB_{i} \bar{\bw}} )$, ${\bB}_{\sf KKT}(\bar{\bw})$ in \eqref{eq:B_KKT} is written as
\begin{align}
    \bB_{\sf KKT}(\bar{\bw}) = \sum_{k=1}^{K}\left(\frac{\bB_{k}}{\bar{\bw}^{\sf H} \bB_{k} \bar{\bw}} \right) \prod_{i=1}^{K} \left( {\bar{\bw}^{\sf H} \bB_{i} \bar{\bw}} \right).
\end{align}
Using \eqref{eq:Bk}, we re-express ${\bB}_{\sf KKT}(\bar{\bw})$ as
\begin{align}
    \label{eq:B_KKT_c_k}
    \bB_{\sf KKT}(\bar{\bw}) = \sum_{k=1}^{K} c_{k} \bB_{k},
\end{align}
where $c_{k} = ( {\bar{\bw}^{\sf H} \bB_{k} \bar{\bw}} )^{-1} \prod_{i=1}^{K} ( {\bar{\bw}^{\sf H} \bB_{i} \bar{\bw}} )$ is a scalar coefficient.
From \eqref{eq:B_KKT_c_k}, we can reformulate ${\bB}_{\sf KKT}(\bar{\bw})$ as a block-diagonal representation, i.e.,
\begin{align}
    \label{eq:B_KKT_rec}
    {\bB}_{\sf KKT} = \mathrm{diag} \left( \tilde{\bB}^{(K)}_{1}, ..., \tilde{\bB}^{(K)}_{K} \right).
\end{align}
Here, each sub-matrix $\tilde{\bB}^{(K)}_{k} \in \bbC^{K \times K}$ is defined in relation to a common representative matrix $\tilde{\bB}^{(K)}$:
\begin{align}
    \label{eq:sub_B_matrix}
    \tilde{\bB}^{(K)}_{k} &= \tilde{\bB}^{(K)} - c_{k} \bH^{\sf H} \bh_{k} \bh^{\sf H}_{k} \bH,
\end{align}
where $\tilde{\bB}^{(K)} \! = \! \sum^{K}_{i=1} c_{i}(\bH^{\sf H} \bh_{i} \bh^{\sf H}_{i} \bH + \frac{\sigma^2}{P} \bH^{\sf H}\bH)$.
Given that ${\bB}_{\sf KKT}$ is block-diagonal, its inverse ${\bB}_{\sf KKT}^{-1}$ requires inverting each of the $K$ sub-matrices $\tilde{\bB}^{(K)}_{k}$.
Inverting all $K$ sub-matrices would cost $\CMcal{O}(K^4)$, which matches the derived complexity of the S-GPIP algorithm.

To reduce this complexity, our strategy involves two steps: first, we compute the inverse of the common matrix, $(\tilde{\bB}^{(K)})^{-1}$, which is a standard matrix inversion with $\CMcal{O}(K^3)$ complexity.
Second, we exploit the rank-one update structure defined in \eqref{eq:sub_B_matrix}.
This structure allows us to apply the Sherman--Morrison formula \cite{golub2013book:Sherman} to efficiently compute the inverse of each $k$th sub-matrix:
\begin{align}
    \label{eq:inv_B_KKT_SM_formula}
    \left( \tilde{\bB}^{(K)}_{k} \right)^{-1} &= \left( \tilde{\bB}^{(K)} - c_{k} \bH^{\sf H} \bh_{k} \bh^{\sf H}_{k} \bH \right)^{-1} \\
    &= \left( \tilde{\bB}^{(K)} \right)^{-1} + \frac{ c_{k} \left( \tilde{\bB}^{(K)} \right)^{-1} \bH^{\sf H} \bh_{k} \bh^{\sf H}_{k} \bH \left( \tilde{\bB}^{(K)} \right)^{-1}}{1 - c_{k} \bh^{\sf H}_{k} \bH \left( \tilde{\bB}^{(K)} \right)^{-1} \bH^{\sf H} \bh_{k}}. \nonumber
\end{align}
Accordingly, we calculate the inverse of $\bB_{\sf KKT}$ as
\begin{align}
    \label{eq:inv_B_KKT}
    {\bB}_{\sf KKT}^{-1} = \mathrm{diag} \left( \left( \tilde{\bB}^{(K)}_{1} \right)^{-1}, ..., \left( \tilde{\bB}^{(K)}_{K} \right)^{-1} \right).
\end{align}
This process involves one initial $\CMcal{O}(K^3)$ inversion to find $(\tilde{\bB}^{(K)})^{-1}$, followed by $K$ applications of the Sherman--Morrison formula in \eqref{eq:inv_B_KKT_SM_formula}.
Since each Sherman--Morrison application involves only matrix-vector products, it costs only $\CMcal{O}(K^2)$.
Therefore, the total per-iteration complexity to compute all $K$ inverses of $\tilde{\bB}^{(K)}_{k}$ in \eqref{eq:inv_B_KKT} is reduced to $\CMcal{O}(K^3 + K \cdot K^2) = \CMcal{O}(K^3)$.
This optimization reduces the total computational complexity of the proposed S-GPIP from $\CMcal{O}(TK^4 + NK^2)$ to $\CMcal{O}(TK^3 + NK^2)$.

\section{Scalable Precoding with Imperfect CSIT}
\label{sec:S-GPI-imperfect}
We now proceed to design the precoding algorithm for the imperfect CSIT scenario.
The estimate of the channel $\bh_{k}\in \bbC^{N}$ is modeled as
\begin{align}
    \label{eq:channel_error}
    {\hat{\bh}}_{k} = \bh_{k} - {\boldsymbol{\phi}}_{k},
\end{align}
where ${\hat{\bh}}_{k}$ is the $k$th column of ${\hat{\bH}} = [{\hat{\bh}}_{1}, ..., {\hat{\bh}}_{K}] \in \bbC^{N \times K}$ and ${\boldsymbol{\phi}}_{k} \in \bbC^{N}$ is a channel estimation error vector.
Specifically, we assume that the BS has access to the estimated channel $\hat{\bf h}_k$ and the error covariance matrix ${\bf{\Phi}}_{k} = \bbE \left[{\boldsymbol{\phi}}_{k} {\boldsymbol{\phi}}^{\sf H}_{k} \right]$, $\forall k$.

Since the BS does not have the perfect CSIT, it is unable to accurately obtain the SE expression in \eqref{eq:rate}.
To address this uncertainty, we adopt the average SE metric \cite{joudeh2016TCOM:WMMSE-SAA}, which evaluates the expected SE over the CSIT error distribution for a given channel estimate.
Thus, based on \eqref{eq:rate}, the average SE of the $k$th user is computed and lower bounded as
\begin{align}
    \nonumber
    &\mathbb{E}[{R}_{k}(\bF)] = \bbE \! \left[ \bbE \! \left[ \log_{2} \! \left(1 + \frac{|\bh^{\sf H}_{k} \bff_{k}|^{2}}{\sum_{i=1, i\ne k}^{K} \! |\bh^{\sf H}_{k} \bff_{i}|^{2} + \frac{\sigma^{2}}{P}} \right) \! \Bigg{|} \hat{\bh}_{k} \right] \right] 
    \\\label{eq:average_SE}
    &\overset{(a)}{\geq} \bbE \! \left[ \log_{2} \! \left(1 + \frac{|\hat{\bh}^{\sf H}_{k} \bff_{k}|^{2}}{\sum_{i=1, i\ne k}^{K} |\hat{\bh}^{\sf H}_{k} \bff_{i}|^{2} \! + \! \sum_{i=1}^{K} \bff^{\sf H}_{i} {\bf{\Phi}}_{k} \bff_{i} + \frac{\sigma^{2}}{P}} \right) \! \right],
\end{align}
where following the generalized mutual information principles \cite{yoo2006TIT:GMI, lapidoth2002TIT:GMI}, $(a)$ is derived by modeling all error terms as independent Gaussian noise, after which Jensen's inequality is invoked to establish the final expression.

From \eqref{eq:average_SE}, we obtain the lower bound problem for the ergodic sum SE problem:
\begin{align}
    \label{eq:inequality}
    \underset{\bF}{\text{maximize}} \sum_{k=1}^{K} \mathbb{E}[{R}_{k}(\bF)] \geq \underset{\bF}{\text{maximize}} \sum_{k=1}^{K} \mathbb{E}[\bar{R}^{\sf lb}_{k}(\bF)],
\end{align}
where
\begin{align}
    \label{eq:rate_lb}
    \bar{R}^{\sf lb}_{k}(\bF) \!=\! \log_{2} \!\left(\! 1 + \frac{|\hat{\bh}^{\sf H}_{k} \bff_{k}|^{2}}{\sum_{i=1, i\ne k}^{K} |\hat{\bh}^{\sf H}_{k} \bff_{i}|^{2} + \! \sum_{i=1}^{K} \bff^{\sf H}_{i} {\bf{\Phi}}_{k} \bff_{i} + \frac{\sigma^{2}}{P}} \! \right).
\end{align}
Since the channel estimate $\hat {\bf h}_k$ is available, we solve the following problem without averaging over the estimated channel to achieve better optimization by exploiting the channel estimates for each fading block:
\begin{subequations}
    \label{avg_problem}
    \begin{align}
        \underset{\bF}{\text{maximize}}&\;\; \sum_{k=1}^{K} \bar{R}^{\sf lb}_{k}(\bF) \label{eq:avg_problem}\\
        \text{subject to}&\;\; \mathrm{tr} \left( \bF \bF^{\sf H} \right) \leq 1. \label{eq:avg_constraint}
    \end{align}
\end{subequations}
In other words, instead of directly solving \eqref{problem} using the channel estimate $\hat {\bf h}_k$,  we solve \eqref{avg_problem}  to achieve robust optimization by explicitly accounting for the channel estimation error.

\subsection{Low-dimensional Projection}
As in Proposition~\ref{proposition1}, we derive the following proposition for nontrivial stationary points of the problem in \eqref{avg_problem} under imperfect CSIT:
\begin{proposition}
    \label{proposition2}
    A nontrivial stationary point $\bff_{k}$ of the problem in \eqref{avg_problem} must lie in the subspace spanned by the columns of the estimated channel $\hat{\bH}$ and channel error covariance matrices ${{\boldsymbol{\Phi}}} \! = \! \left[{\boldsymbol{\Phi}}_{1},...,{\boldsymbol{\Phi}}_{K} \right]$ for $k \in \{1,...,K\}$.
    Thus, $\bff_{k}$ can be expressed as $\bff_{k} \! = \! \bG \bv_{k}$, where $\bG \! = \! \left[\hat{\bH}, {{\boldsymbol{\Phi}}} \right] \! \in \! \bbC^{N \times (K + NK)}$ is the augmented matrix and $\bv_{k} \! \in \!  \bbC^{(K + NK)}$ is a corresponding unique vector.
\end{proposition}
\begin{proof} 
    See Appendix~\ref{app:impCSIT}.
\end{proof}

\begin{remark}
    \label{rm2}
    \normalfont In Proposition~\ref{proposition2}, we note that the channel error covariance matrices significantly expand the dimension of the search space of the optimization problem in massive MIMO systems with imperfect CSIT.
    This expansion makes finding nontrivial stationary points considerably more challenging than under the perfect CSIT assumption of Proposition~\ref{proposition1}.
\end{remark}
As explained in Remark~\ref{rm2}, the imperfect CSIT assumption complicates identifying a nontrivial stationary point because channel estimation errors expand the search space.
To maintain computational tractability, we constrain the beamforming design to be restricted to the subspace spanned by the dominant eigenvectors of each user's error covariance matrix.
Accordingly, the dimensionality of this subspace is determined by the rank of $\boldsymbol{\Phi}_{k}$, leading to a low-rank approximation of the error covariance matrix for user $k$:
\begin{align}
    \label{eq:power_iteration_imperfect}
    {\boldsymbol{\Phi}}_{k} \approx \sum_{i=1}^{r_{k}} \lambda_{i,k} {{\boldsymbol{\psi}}}_{i,k}{{\boldsymbol{\psi}}}_{i,k}^{\sf H},
\end{align}
where $r_{k}$ is a positive integer value, which satisfies $0 \! < \! r_{k} \! \leq \! \mathrm{rank}(\boldsymbol{\Phi}_{k})$, and $\lambda_{i,k}$ and ${\boldsymbol{\psi}}_{i,k}$ are eigenvalues and eigenvectors of ${\boldsymbol{\Phi}}_{k}$, respectively.
For analytical simplicity, we assume $\lambda_{1,k}\geq...\geq \lambda_{r_{k},k}$ and $r = r_{k}$, i.e., all $r_{k}$ values are the same.

From \eqref{eq:power_iteration_imperfect}, the precoding vector for the $k$th user is defined as
\begin{align}
\label{eq:reformulated_opt_precoding_vector_imperfect}
    \bff_{k} = \hat{\bG} {\bv}_{k},
\end{align}
where $\hat{\bG} \! = \! [\hat{\bH}, {\boldsymbol{\Psi}}] \! \in \! \bbC^{N \times (K + rK)}$ is a combined matrix of the estimated channel $\hat{\bH}$ and the aggregated matrix of dominant eigenvectors ${\boldsymbol{\Psi}} \! = \! [{{\boldsymbol{\Psi}}}_{1}, ...,{{\boldsymbol{\Psi}}}_{K}] \! \in \! \bbC^{N \times rK}$.
Here, each block ${{\boldsymbol{\Psi}}}_{k} = [{{\boldsymbol{\psi}}}_{1,k}, ...,{{\boldsymbol{\psi}}}_{r,k}]$ collects the $r$ dominant eigenvectors of $\boldsymbol{\Phi}_{k}$. 
Thus, we define the precoder as a function of the estimated channel and approximated channel error covariance matrices:
\begin{align}
    \label{eq:opt_precoder_imp}
    \bF = \hat{\bG} {\bV},
\end{align}
where ${\bV} \! = \! [{\bv}_{1}, ..., {\bv}_{K}] \! \in \! \bbC^{(K+rK) \times K}$ is a unique weight matrix under the approximation in \eqref{eq:power_iteration_imperfect}.
As in \eqref{eq:normalization}, applying the normalization ${\bF}/{\|\bF\|_{\sf F}} \! = \! {\hat{\bG} {\bV}}/{\sqrt{\mathrm{tr} (\hat{\bG} {\bV} {\bV}^{\sf H} \hat{\bG}^{\sf H} )}}$, we reformulate the lower bound on the SE for the $k$th user in \eqref{eq:rate_lb} as
\begin{align}
    \label{eq:re_rate_imperfect}
    &\bar{R}^{\sf lbRef}_{k}({\bV}) = \log_{2} \left(1 + \frac{|\hat{\bh}^{\sf H}_{k} \hat{\bG} {\bv}_{k}|^{2}}{\sum_{i=1, i\ne k}^{K} |\hat{\bh}^{\sf H}_{k} \hat{\bG} {\bv}_{i}|^{2} + {\phi}_{k}^{\sf lb} + {z}^{\sf lb}} \right),
\end{align}
where ${\phi}_{k}^{\sf lb} \! = \! \sum_{i=1}^{K} {\bv}^{\sf H}_{i} \hat{\bG}^{\sf H} \boldsymbol{\Phi}_{k} \hat{\bG} {\bv}_{i}$ and ${z}^{\sf lb} \! = \! \frac{\sigma^{2}}{P} \mathrm{tr} (\hat{\bG} {\bV} {\bV}^{\sf H} \hat{\bG}^{\sf H} )$.

Accordingly, the optimization problem \eqref{avg_problem} is reformulated as
\begin{subequations}
    \label{re_problem2}
    \begin{align}
        \underset{{\bV}}{\text{maximize}}&\;\; \sum_{k=1}^{K} \bar{R}^{\sf lbRef}_{k}({\bV}) \label{eq:re_problem2}\\
        \text{subject to}&\;\; \mathrm{tr} \left(\hat{\bG} {\bV} {\bV}^{\sf H} \hat{\bG}^{\sf H} \right) = 1.
        \label{eq:re_constraint2}
    \end{align}
\end{subequations}
The approximation in \eqref{eq:reformulated_opt_precoding_vector_imperfect} is inherently suboptimal relative to the comprehensive framework of Proposition~\ref{proposition2}, since it does not fully span the space defined by all channel error covariance matrices.
Nevertheless, by constraining the optimization to the subspace formed by the $r$ dominant eigenvectors of each $\boldsymbol{\Phi}_{k}$, the dimensionality of the problem becomes independent of the number of BS antennas $N$.
This approach facilitates a substantial reduction in computational complexity in exchange for a marginal degradation in SE performance.

\subsection{Reformulation and Algorithm}
Similar to the perfect CSIT case, we apply the GPIP approach.
To this end, we transform the reformulated problem in \eqref{re_problem2} into the GPIP-friendly form.
Let $\bar{\bN}_{k}$ be
\begin{align}
    \label{eq:Mk_imperfect}
    \bar{\bN}_{k} = \hat{\bG}^{\sf H} \hat{\bh}_{k} \hat{\bh}_{k}^{\sf H} \hat{\bG} + \hat{\bG}^{\sf H} \boldsymbol{\Phi}_{k} \hat{\bG} \in \bbC^{(K+rK) \times (K+rK)}.
\end{align}
As in \eqref{eq:re_normalization}, by stacking ${\bv}_{k}$ as $\bar{\bv} \! = \! [{\bv}_{1}^{\sf T}, ..., {\bv}_{K}^{\sf T}]^{\sf T} \! \in \! \bbC^{(K^{2} + rK^{2})}$, we rewrite the trace term in $z^{\sf lb}$ as
\begin{align}
    \label{eq:re_normalization_imp}
    \mathrm{tr} \left(\hat{\bG} {\bV} {\bV}^{\sf H} \hat{\bG}^{\sf H} \right) = \bar{\bv}^{\sf H} \left(\bI_{K} \otimes \hat{\bG}^{\sf H} \hat{\bG} \right) \bar{\bv}.
\end{align}
Using \eqref{eq:Mk_imperfect} and \eqref{eq:re_normalization_imp}, we represent the lower bound SE expression $\bar{R}^{\sf lbRef}_{k}({\bV})$ as the Rayleigh quotient form:
\begin{align}
    \label{eq:rayleigh_form_imperfect}
    \bar{R}^{\sf lbRef}_{k}(\bar{\bv}) = \mathrm{log}_{2} \left(\frac{\bar{\bv}^{\sf H} \bar{\bA}_{k} \bar{\bv}}{\bar{\bv}^{\sf H} \bar{\bB}_{k} \bar{\bv}} \right),
\end{align}
where $\bar{\bA}_{k}$ and $\bar{\bB}_{k}$ are obtained as
\begin{align}
    &\bar{\bA}_{k} \! = \bI_{K} \otimes \bar{\bN}_{k} + \frac{\sigma^{2}}{P} \left(\bI_{K} \otimes \hat{\bG}^{\sf H} \hat{\bG} \right) \! \in \! \bbC^{(K^{2}+rK^{2}) \times (K^{2}+rK^{2})},
    \label{Ck_imperfect}
    \\
    &\bar{\bB}_{k} \! = \bar{\bA}_{k} - \be_{k}\be_{k}^{\sf H} \otimes \hat{\bG}^{\sf H} \hat{\bh}_{k} \hat{\bh}_{k}^{\sf H} \hat{\bG} \in \bbC^{(K^{2}+rK^{2}) \times (K^{2}+rK^{2})}.
    \label{Dk_imperfect}
\end{align}

From \eqref{eq:rayleigh_form_imperfect}, the optimization problem \eqref{avg_problem} is expressed as
\begin{align}
    \label{eq:re_problem_imperfect}
    \underset{\bar{\bv}}{\text{maximize}}&\;\; \sum_{k=1}^{K} \bar{R}^{\sf lbRef}_{k}(\bar{\bv}).
\end{align}
As in \eqref{eq:re_problem}, the objective function in \eqref{eq:re_problem_imperfect} is scale-invariant
with respect to $\bar{\bv}$ and the reformulated problem remains non-convex.
For \eqref{eq:re_problem_imperfect}, we also identify the stationary condition as in Lemma~\ref{Lemma2}.
\begin{lemma}
    \label{Lemma2}
    The first-order optimality condition of \eqref{eq:re_problem_imperfect} is represented in the generalized eigenvalue problem as
    \begin{align}
        \bar{\bB}_{\sf KKT}^{-1}(\bar{\bv})\bar{\bA}_{\sf KKT}(\bar{\bv}) = \lambda(\bar{\bv})\bar{\bv},
        \label{eq:NEPv_imp}
    \end{align}
    where
    \begin{align}
        \bar{\bA}_{\sf KKT}(\bar{\bv}) &= \sum_{k=1}^{K}\left(\frac{\bar{\bA}_{k}}{\bar{\bv}^{\sf H} \bar{\bA}_{k} \bar{\bv}} \right) \lambda_{\sf num}(\bar{\bv}),
        \label{eq:A_KKT_imp}
        \\
        \bar{\bB}_{\sf KKT}(\bar{\bv}) &= \sum_{k=1}^{K}\left(\frac{\bar{\bB}_{k}}{\bar{\bv}^{\sf H} \bar{\bB}_{k} \bar{\bv}} \right) \lambda_{\sf den}(\bar{\bv}),
        \label{eq:B_KKT_imp}
        \\
        \lambda(\bar{\bv})  &= \prod_{k=1}^{K} \left(\frac{\bar{\bv}^{\sf H} \bar{\bA}_{k} \bar{\bv}}{\bar{\bv}^{\sf H} \bar{\bB}_{k} \bar{\bv}} \right).
        \label{eq:eig_val_imp}
    \end{align}
    The numerator $\lambda_{\sf num}(\bar{\bv})$ and denominator $\lambda_{\sf den}(\bar{\bv})$ of the eigenvalue are any functions that satisfy $\lambda(\bar{\bv})=\lambda_{\sf num}(\bar{\bv})/\lambda_{\sf den}(\bar{\bv})$.
    \begin{proof}
        The proof follows the same as in Appendix~\ref{app:Lemma1}.
    \end{proof}
\end{lemma}

\begin{algorithm}[!t]
\caption{Scalable GPIP with Error Covariance under Imperfect CSIT}\label{alg:two}
\bf{initialize}: $\bar{\bv}^{(0)}\gets$ \rm{RZF}\\
\rm{Set the iteration count} $t=1.$\\
\While{$\|\bar{\bv}^{(t)} - \bar{\bv}^{(t-1)} \| > \varepsilon_{\bar{\bv}}$ \rm{and} $t \leq t_{\sf max}$}{
\rm{Build matrix} $\bar{\bA}_{\sf KKT}(\bar{\bv}^{(t-1)})$ according to \eqref{eq:A_KKT_imp}.\\
\rm{Build matrix} $\bar{\bB}_{\sf KKT} (\bar{\bv}^{(t-1)})$ according to \eqref{eq:B_KKT_imp}.\\
\rm{Compute} $\bar{\bv}^{(t)}$ = $\bar{\bB}_{\sf KKT}^{-1} (\bar{\bv}^{(t-1)}) \bar{\bA}_{\sf KKT} (\bar{\bv}^{(t-1)}) \bar{\bv}^{(t-1)}.$\\
Normalize $\bar{\bv}^{(t)} ={\bar{\bv}^{(t)}}/{\|\bar{\bv}^{(t)}\|}.$\\
$t \gets t+1.$\\
}
$\bar{\bv}^{\star} \gets \bar{\bv}^{(t)}.$\\
\rm{Build the unique weight matrix} ${\bV}^{\star} \gets \bar{\bv}^{\star}$.\\
\rm{Compute} $\bF = \hat{\bG} {\bV}^{\star}$.\\
\rm{Normalize} $\bF^{\star} = \bF/\|\bF\|_{\sf{F}}$.\\
\Return $\bF^{\star}$
\end{algorithm}

The proposed S-GPIP with error covariance (S-GPIP (cov)) method is summarized in Algorithm~\ref{alg:two}.
Since the optimization structure is isomorphic to that of Algorithm~\ref{alg:one}, we omit the repetitive iterative details.
Algorithm~\ref{alg:two} follows the same procedure but operates on the stacked weight vector $\bar{\bv}^{(t)}$, updating the matrices $\bar{\bA}_{\sf KKT} (\bar{\bv}^{(t-1)})$ and $\bar{\bB}_{\sf KKT} (\bar{\bv}^{(t-1)})$ according to \eqref{eq:A_KKT_imp} and \eqref{eq:B_KKT_imp}.
Upon convergence, the final precoder $\bF^{\star}$ is reconstructed from the optimized weight vector $\bar{\bv}^{\star}$ and normalized.

\subsection{Complexity Reduction}
\label{subsec:Complexity_Analysis2}

For the imperfect CSIT case, the computational complexity of Algorithm~\ref{alg:two} is also primarily determined by the matrix inversion $\bar{\bB}_{\sf KKT}^{-1} ( \bar{\bv}^{(t-1)} )$.
Our method utilizes a $(K^2+rK^2) \times (K^2+rK^2)$ block-diagonal matrix $\bar{\bB}_{\sf KKT}(\bar{\bv})$.
The inversion complexity per-iteration is thus $\CMcal{O}(K\cdot (K+rK)^{3}) = \CMcal{O}(\bar{r}^{3}K^4)$, where $\bar{r}=r+1$.
The total complexity, including $T$ iterations and the $\CMcal{O}(\bar{r}^2NK^2)$ complexity for pre-processing (e.g., finding $r$ eigenvectors for each user and computing $\hat{\bG}^{\sf H}\hat{\bG}$), is $\CMcal{O}(T\bar{r}^3 K^4 +  \bar{r}^2NK^2)$.
Consequently, by appropriately selecting $r \ll N$, the proposed S-GPIP (cov) algorithm achieves both high SE performance and low computational burden in massive MIMO systems under imperfect CSIT.

Similar to the perfect CSIT case, we can further reduce the inversion complexity to $\CMcal{O}(\bar{r}^3 K^3)$ per iteration by using the Sherman--Morrison formula-based computation. 
This optimization reduces the total computational complexity of the error covariance-based S-GPIP, S-GPIP (cov), from $\CMcal{O}(T\bar{r}^3 K^4 + \bar{r}^2 NK^2)$ to $\CMcal{O}(T\bar{r}^3 K^3 + \bar{r}^2 NK^2)$.

Table~\ref{Table1} summarizes the complexity order of GPIP and proposed S-GPIP. 
The proposed S-GPIP algorithms has a linear complexity with respect to the number of antennas $N$, enabling scalability to large antenna arrays. 
Therefore, the proposed methods is well suited for implementation in massive MIMO systems where $N\gg K$.

\begin{table}[t]
    \caption{Algorithm Complexity ($\bar{r}=r+1$)}
    \label{Table1}
    \centering
    {
    \scriptsize
        \begin{tabular}{c|c|c}
        \hline
        \hline
         {Algorithms} & {Naive Complexity Order} & {Reduced Complexity Order} \\ \cline{1-3} 
         {GPIP} & {$\CMcal{O}(TN^3 K)$} & {$\CMcal{O}(TN^2 K)$} \\
         {GPIP (cov)} & {$\CMcal{O}(T N^3 K)$} & {$\CMcal{O}(T \bar{r} N^2 K)$} \\
         {S-GPIP} & {$\CMcal{O}(TK^4 + NK^2)$} & {$\CMcal{O}(TK^3 + NK^2)$} \\
         {S-GPIP (cov)} & {$\CMcal{O}(T\bar{r}^3K^4 + \bar{r}^2NK^2)$} & {$\CMcal{O}(T\bar{r}^3K^3 + \bar{r}^2NK^2)$} \\
         \hline
         \hline
        \end{tabular}
    }
\end{table}






\section{Convergence}
\label{sec:conv_pg}
In this section, we interpret the GPIP update rule as a class of PPGA methods and establish its convergence properties under perfect CSIT. The analysis naturally extends to the imperfect CSIT case.

\subsection{PPGA Interpretation of the GPIP Update Rule}
We consider the sum-log objective in \eqref{eq:re_problem} as
\begin{align}
\label{eq:obj_L}
    L(\bar{\bw}) = \sum_{k=1}^{K}\log_{2} \left( \frac{\bar{\bw}^{\sf H}{\bA}_k\bar{\bw}}{\bar{\bw}^{\sf H}{\bB}_k\bar{\bw}} \right),
\end{align}
which is scale-invariant, i.e., $L(\alpha \bar{\bf w})=L(\bar{\bf w})$ for any $\alpha\neq 0$. Hence, we impose the unit-sphere constraint as
\begin{align}
\label{eq:sphere}
    \mathcal{S} = \{\bar{\bf w}\in\mathbb C^{K^2}:\ \|\bar{\bf w}\|_2=1\}.
\end{align}
Now, let us define the matrices
\begin{align}
\label{eq:tildeAtildeB}
    \widetilde{{\bA}}(\bar{\bw}) = \sum_{k=1}^{K} \left( \frac{{\bA}_k}{\bar{\bw}^{\mathsf H}{\bA}_k\bar{\bw}} \right),
    \qquad
    \widetilde{{\bB}}(\bar{\bw}) = \sum_{k=1}^{K} \left( \frac{{\bB}_k}{\bar{\bw}^{\mathsf H}{\bB}_k\bar{\bw}} \right).
\end{align}
A standard complex matrix derivative yields
\begin{align}
\label{eq:grad_L}
    \nabla L(\bar{\bw}) = 2\left( \widetilde{{\bA}}(\bar{\bw})-\widetilde{{\bB}}(\bar{\bw}) \right) \bar{\bw}.
\end{align}
We define the GPIP mapping before the normalization (line 6 in Algorithm~\ref{alg:one}) as
\begin{align}
\label{eq:g_mapping}
    {\bg}(\bar{\bw}) = \widetilde{{\bB}}^{-1}(\bar{\bw})\widetilde{{\bA}}(\bar{\bf w})\bar{\bw}.
\end{align}
We remark that indeed, $\lambda_{\sf num}(\bar{\bw})$ and $\lambda_{\sf den}(\bar{\bw})$ can be omitted when implementing the GPIP algorithm since they contribute only to the scale of the updated vector, which does not affect the update due to the normalization step in the algorithm.
Combining \eqref{eq:grad_L} and \eqref{eq:g_mapping}, we obtain the exact identity
\begin{align}
    \label{eq:g_equals_grad}
    {\bg}(\bar{\bw}) = \bar{\bw} + \tfrac{1}{2} \, \widetilde{{\bB}}^{-1}(\bar{\bw})\nabla L(\bar{\bw}).
\end{align}
Therefore, ${\bg}(\bar{\bw}) - \bar{\bw}$ is proportional to a preconditioned gradient direction.

The Euclidean projection onto $\mathcal S$ is
\begin{align}
    \label{eq:proj_sphere}
    \mathrm{Proj}_{\mathcal S}({\bz})=\frac{{\bz}}{\|{\bz}\|_2}.
\end{align}
Motivated by \eqref{eq:g_equals_grad}, we consider the following  PPGA update:
\begin{align}
\label{eq:ppg_update}
    \bar{\bw}^{t+1}
    =\mathrm{Proj}_{\mathcal S}\! \left( \bar{\bw}^t+\eta\,{\bf P}(\bar{\bw}^t) \nabla L(\bar{\bw}^t) \right),
\end{align}
where $\eta>0$ is a step size and 
\begin{align}
    {\bP}(\bar{\bw}) = \widetilde{{\bB}}^{-1}(\bar{\bw}).
\end{align}
When $\eta=\tfrac{1}{2}$, we have
\begin{align}
    \bar{\bw}^{t+1} = \mathrm{Proj}_{\mathcal S} \left( \bg (\bar{\bw}^t) \right),
\end{align}
since \eqref{eq:g_equals_grad} implies
$\bg(\bar{\bw}^t) = \bar{\bw}^{t} + \frac{1}{2} \bP(\bar{\bw}^t) \nabla{L}(\bar{\bw}^t)$.
Consequently, the proposed damped update coincides with a GPIP step.

\subsection{Convergence Analysis and Convergent S-GPIP}
To prove the convergence of the GPIP through the interpretation as PPGA, we introduce the following assumptions.
\begin{assumption}[Euclidean smoothness near $\mathcal S$]
    \label{ass:smooth}
    The objective $L(\bar{\bw})$ admits a Lipschitz continuous Euclidean gradient on an open neighborhood $\mathcal U$ containing the sphere 
    $\mathcal S = \{\bar{\bw}:\|\bar{\bw}\|_2 = 1\}$, i.e.,
    \begin{align}
        \|\nabla L(\bx)-\nabla L(\by)\|_2
        \le
        L_g \|\bx-\by\|_2,
        \qquad
        \forall \bx,\by \in \mathcal U.
    \end{align}
\end{assumption}

\begin{assumption}[Uniformly bounded preconditioner]
    \label{ass:bounded_precond}
    There exist constants $0 < m \le M < \infty$ such that for all $\bar{\bw}\in\mathcal S$,
    \begin{align}
        \label{eq:bounded_precond}
        m{\bI} \preceq {\bP}(\bar{\bw})\preceq M{\bI}.
    \end{align}
    Equivalently, $\widetilde{{\bB}}(\bar{\bw})$ is uniformly positive definite on $\mathcal S$.
\end{assumption}

In the following lemma, we show that Assumption~\ref{ass:bounded_precond} is satisfied with specific $m$ and $M$.
To this end, we recall that
\begin{align}
    \label{eq:AkBk_structure}
    {\bA}_k &= {\bI}_K\otimes {\bN}_k + {\bC},
    \\
    {\bB}_k &= {\bA}_k - {\be}_k{\be}_k^{\sf H}\otimes {\bN}_k,
    \\
    {\bC} &= \frac{\sigma^2}{P} \left( {\bI}_K \otimes {\bH}^{\sf H}{\bH} \right),
\end{align}
and define ${\bC}_0 = \frac{\sigma^2}{P}{\bf H}^{\sf H}{\bf H}$ so that ${\bC}={\bI}_K \otimes {\bC}_0$.

\begin{lemma}[Uniform bounds for $\widetilde{{\bB}}^{-1}(\bar{\bw})$]
\label{lem:bounded_precond_block}
Assume $\sigma^2>0$ and ${\bH}^{\sf H}{\bH} \succ {\bf 0}$.
Let
\begin{align}
    \label{eq:bmin_bmax_def}
    b_{\min} = \lambda_{\min}({\bC}_0),
    \qquad
    b_{\max} = \max_{k}\lambda_{\max}({\bN}_k+{\bC}_0),
\end{align}
and define block matrices
\begin{align}
    \label{eq:S_i_def}
    {\bS}^{(i)} = \left( \sum_{k=1, k \neq i}^{K} {\bN}_k \right) + K{\bC}_0,\qquad i=1,\dots,K.
\end{align}
Then, for all $\bar{\bf w}\in\mathcal S$,
\begin{align}
    \label{eq:tildeB_bounds}
    \frac{1}{b_{\max}} \left( \sum_{k=1}^{K}{\bB}_k \right)\ \preceq\ \widetilde{{\bB}}(\bar{\bw})\ \preceq\ \frac{1}{b_{\min}}\left( \sum_{k=1}^{K}{\bB}_k \right),
\end{align}
and consequently
\begin{align}
    \label{eq:P_bounds_tight}
    m{\bI} \preceq \widetilde{{\bB}}^{-1}(\bar{\bw})\preceq M{\bI},\ \forall \bar{\bw}\in\mathcal S,
\end{align}
where
\begin{align}
    m = \frac{b_{\min}}{\max_i\lambda_{\max}({\bS}^{(i)})},
    \qquad
    M = \frac{b_{\max}}{\min_i\lambda_{\min}({\bS}^{(i)})}.
\end{align}
\end{lemma}

\begin{proof}
From \eqref{eq:AkBk_structure}, each ${\bf B}_k$ is block-diagonal with the $i$th block equal to ${\bf C}_0$ if $i=k$, and ${\bf N}_k+{\bf C}_0$ if $i\neq k$.
Thus, $\lambda_{\min}({\bB}_k) \ge \lambda_{\min}({\bC}_0) = b_{\min} > 0$ and
$\lambda_{\max}({\bB}_{k}) = \lambda_{\max}({\bN}_{k} + {\bC}_0)\le b_{\max}$.
For $\|\bar{\bw}\|_2 = 1$, this implies $b_{\min} \le \bar{\bw}^{\sf H}{\bB}_{k} \bar{\bw} \le b_{\max}$, so that
\begin{align}
    \frac{1}{b_{\max}}{\bf B}_{k} \preceq \frac{{\bf B}_{k}}{\bar{\bf w}^{\mathsf H} {\bf B}_{k} \bar{\bf w}} \preceq \frac{1}{b_{\min}}{\bf B}_{k}.
\end{align}
Summing over $k$ yields \eqref{eq:tildeB_bounds}.
Moreover, $\sum_{k=1}^{K} {\bB}_{k}$ is block-diagonal with the $i$th block given by ${\bS}^{(i)}$ in \eqref{eq:S_i_def}, therefore $\lambda_{\max}(\sum_k {\bB}_k) = \max_{i} \lambda_{\max}({\bS}^{(i)})$ and
$\lambda_{\min}(\sum_{k=1}^{K} {\bB}_{k})=\min_{i} \lambda_{\min}({\bS}^{(i)})$.
Applying eigenvalue inequalities and inversion completes \eqref{eq:P_bounds_tight}.
\end{proof}

\begin{theorem}
\label{thm:ppg_convergence}
Suppose Assumptions~\ref{ass:smooth}--\ref{ass:bounded_precond} hold.
Consider the iteration \eqref{eq:ppg_update} with a constant step size $\eta$ satisfying
\begin{align}
    \label{eq:eta_condition}
    0 < \eta < \frac{2m}{L_{g} M^2}.
\end{align}
Then:
\begin{enumerate}
\item (\emph{Sufficient ascent}) The objective is nondecreasing along the iterates:
\begin{align}
    \label{eq:monotone}
    L(\bar{\bw}^{t+1}) \ge L(\bar{\bw}^{t}), \quad \forall t.
\end{align}
\item (\emph{Vanishing gradient}) We have $\|\nabla L(\bar{\bw}^{t})\|_2\to 0$ as $t \to \infty$.
In particular, every accumulation point $\bar{\bw}^\star$ of $\{\bar{\bw}^{t}\}$ is a stationary point of \eqref{eq:obj_L} on $\mathcal S$.
\end{enumerate}
\end{theorem}

\begin{proof}
Let ${\bz}^{t+1} = \bar{\bw}^{t} + \eta {\bP}(\bar{\bw}^{t})\nabla L(\bar{\bw}^{t})$ and $\bar{\bw}^{t+1} = \mathrm{Proj}_{\mathcal S}({\bz}^{t+1})$.
We assume that the step size $\eta$ is sufficiently small so that $\bz^{t+1} \in \mathcal{U}$ is guaranteed.
By $L_g$-smoothness ascent lemma \cite{nesterov2018lectures} applied to $L(\bar{\bw})$,
\begin{align}
    \label{eq:ascent_lemma_step}
    L({\bz}^{t+1}) \ge& L(\bar{\bw}^t) + \eta \left\langle \nabla L(\bar{\bw}^{t}),\,{\bP}(\bar{\bw}^{t}) \nabla L(\bar{\bw}^t) \right\rangle \nonumber
    \\
    &- \frac{L_g}{2} \eta^2 \|{\bP}(\bar{\bw}^{t}) \nabla L(\bar{\bw}^{t})\|_2^2.
\end{align}
Using Assumption~\ref{ass:bounded_precond}, we have the following inequalities:
\begin{align}
    \label{eq:quad_bounds_1}
    \left\langle \nabla L(\bar{\bw}^{t}),\,{\bP (\bar{\bw}^{t})} \nabla L(\bar{\bw}^{t}) \right\rangle &\ge m \|\nabla L(\bar{\bw}^{t}) \|_2^2,
    \\
    \label{eq:quad_bounds_2}
    \|{\bP (\bar{\bw}^{t})}\nabla L(\bar{\bw}^{t}) \|_2^2 &\le M^2\|\nabla L(\bar{\bw}^{t})\|_2^2.
\end{align}
Substituting \eqref{eq:quad_bounds_1} and \eqref{eq:quad_bounds_2} into \eqref{eq:ascent_lemma_step} yields
\begin{align}
    \label{eq:increase_lower_bound}
    L({\bz}^{t+1})-L(\bar{\bw}^t) \ge \left(\eta m-\frac{L_g}{2} \eta^2 M^2 \right)\|\nabla L(\bar{\bw}^t)\|_2^2.
\end{align}
Under \eqref{eq:eta_condition}, the coefficient is nonnegative and strictly positive whenever $\nabla L(\bar{\bw}^t)\neq {\bf 0}$.
Since $L$ is scale-invariant, normalizing ${\bz}^{t+1}$ onto $\mathcal S$ does not change the objective value, i.e.,
\begin{align}
    \label{eq:Lw_Lz}
    L(\bar{\bw}^{t+1})=L({\bz}^{t+1}),
\end{align}
which implies \eqref{eq:monotone}.
Moreover, since the iterates $\{\bar{\bw}^t\}$ lie on the compact set 
$\mathcal S = \{\bar{\bw}:\|\bar{\bw}\|_2 = 1\}$ 
and $L$ is continuous on a neighborhood of $\mathcal S$, 
the sequence $\{L(\bar{\bw}^t)\}$ is bounded above and therefore convergent.
Summing \eqref{eq:increase_lower_bound} over $t$ then gives $\sum_t \|\nabla L(\bar{\bw}^t)\|_2^2 < \infty$, implying
$\|\nabla L(\bar{\bw}^t)\|_2 \to 0$ and that any accumulation point is stationary on $\mathcal S$.
\end{proof}

\begin{remark}[Convergent S-GPIP update]
\label{rem:eta_one}
\normalfont The exact S-GPIP step corresponds to a  move toward ${\bg}(\bar{\bw})$ with $\eta = 1/2$. 
Such a fixed $\eta$, however, may violate the convergence condition in
Theorem~\ref{thm:ppg_convergence}.
To guarantee the convergence, motivated by Theorem~\ref{thm:ppg_convergence}, we can use the damped update in \eqref{eq:ppg_update} with adapting $\eta$ by employing a backtracking line-search to ensure the convergence condition \eqref{eq:eta_condition}, resulting in a convergent S-GPIP algorithm.
\end{remark}
The convergent S-GPIP algorithm is described in Algorithm~\ref{alg:CS-GPIP}. 
The algorithm is applicable to both the perfect CSIT and imperfect CSIT S-GPIP cases by defining $L(\cdot)$, $\widetilde{\bA}(\cdot)$, and $\widetilde{\bB}(\cdot)$ according to each case. 
Since the complexity is also determined by the inversion of $\widetilde{\bB}(\cdot)$, the overall complexity remains the same.

\begin{algorithm}[!t]
\caption{Convergent Scalable GPIP}\label{alg:CS-GPIP}
\bf{initialize}: $\bar{\bu}^{(0)}\gets$ \rm{RZF}\\
\rm{Set the iteration count} $t=1.$\\
\While{$\|\bar{\bu}^{(t)} - \bar{\bu}^{(t-1)} \| > \varepsilon_{\bar{\bu}}$ \rm{and} $t \leq t_{\sf max}$}{
\rm{Build } $\widetilde{\bA}(\bar{\bu}^{(t-1)})$ according to \eqref{eq:A_KKT} or \eqref{eq:A_KKT_imp}.\\
\rm{Build } $\widetilde{\bB}(\bar{\bu}^{(t-1)})$ according to  \eqref{eq:B_KKT} or \eqref{eq:B_KKT_imp}.\\
\rm{Compute} $\eta$ using backtracking from \eqref{eq:ppg_update}\\
\rm{Compute} $  \bar{\bf u}^{(t)}
    =\bar{\bf u}^{(t-1)}+\eta\,\widetilde{\bB}^{-1}(\bar{\bu}^{(t-1)})\nabla L(\bar{\bf u}^{(t-1)}),$\\
Normalize $\bar{\bu}^{(t)} ={\bar{\bu}^{(t)}}/{\|\bar{\bu}^{(t)}\|}.$\\
$t \gets t+1.$\\
}
$\bar{\bu}^{\star} \gets \bar{\bu}^{(t)}.$\\
\rm{Build the unique weight matrix} ${\bU}^{\star} \gets \bar{\bu}^{\star}$.\\
\rm{Compute} $\bF = {\bG} {\bU}^{\star}$ or $\bF = \hat{\bG} {\bU}^{\star}$.\\
\rm{Normalize} $\bF^{\star} = \bF/\|\bF\|_{\sf{F}}$.\\
\Return $\bF^{\star}$
\end{algorithm}

\section{Numerical Results}
\label{sec:Numerical_results}
In this section, we evaluate the SE performance, computational complexity, and convergence of the proposed methods.
For simulations, we assume that the channel vector follows a complex Gaussian distribution $\bh_{k} \sim \cC\cN({\bf{0}}_{N \times 1},\bK_{\bh_{k}})$, where $\bK_{\bh_{k}}$ is modeled by a geometric one-ring scattering channel via the Karhunen-Loeve model \cite{adhikary2013TIT:one-ring, jiang2015TWC:one-ring}.
In frequency division duplex (FDD) systems, the BS typically acquires CSIT through limited feedback.
We model ${\hat{\bh}}_{k}$ as the MMSE estimate of $\bh_{k}$. 
According to the standard MMSE estimation theory for Gaussian random vectors, the estimation error vector ${\boldsymbol{\phi}}_{k}$ is a zero mean vector with covariance matrix ${\bf{\Phi}}_{k} \! = \! \bK_{\bh_{k}} \! - \bK_{\hat{\bh}_{k}}$, where $\bK_{\hat{\bh}_{k}}$ is the covariance of the estimate ${\hat{\bh}}_{k}$.
Using the Karhunen-Loeve model, the channel vector is decomposed as $\bh_{k} \! = \! \bU_{k} {\bf{\Lambda}}^{1/2}_{k} \bg_{k}$, where ${\bf{\Lambda}}_{k}$ is a diagonal matrix which contains the non-zero eigenvalues of $\bK_{\bh_{k}}$, $\bU_{k}$ is a matrix of corresponding eigenvectors, and {$\bg_{k} \! \sim \! \cC\cN({\bf{0}}_{N},\bI_{N})$} is a random vector with IID entries.
Accordingly, the estimated channel for user $k$ is described as
\begin{align}
    \label{eq:estimated_channel}
    \hat{\bh}_{k} = \bU_{k} {\bf{\Lambda}}^{1/2}_{k} \left( \sqrt{1-\kappa^{2}}\bg_{k} + \kappa \bq_{k} \right),
\end{align}
where $\kappa \in [0,1]$ indicates the quality of instantaneous CSIT and $\bq_{k}$ follows IID {$\cC\cN({\bf{0}}_{N},\bI_{N})$}.
Consequently, the channel estimation error covariance matrix is given by
\begin{align}
    \label{eq:error_covariance}
    {\bf{\Phi}}_{k} = \bbE \left[{\boldsymbol{\phi}}_{k} {\boldsymbol{\phi}}^{\sf H}_{k} \right] = \bU_{k} {\bf{\Lambda}}^{1/2}_{k} \left(2-2\sqrt{1-\kappa^{2}} \right) ({\bf{\Lambda}}^{1/2}_{k})^{\sf H} \bU^{\sf H}_{k}.
\end{align}
By varying $\kappa$, the channel estimation accuracy can be conveniently controlled within this model. Although the proposed S-GPIP algorithm accommodates arbitrary CSIT imperfections, we adopt this FDD-based channel estimation model for simplicity.

Unless specified otherwise, the common simulation parameters are as follows: all users are uniformly distributed at distances $d \in [20, 100]\mathrm{\,m}$ from the BS.
The convergence thresholds are set to $\epsilon_{\bar{\bw}} = \epsilon_{\bar{\bv}} = \epsilon_{\bar{\bu}} = 10^{-2}$, with a maximum of $t_{\sf max} = 100$ iterations.
The large-scale fading follows the 3GPP Urban Microcell (UMi) NLoS path loss model defined in \cite{tr385g:3GPP_Re18}: $PL(f_c,d) \! = \! 35.3 \log_{10}(d_{\sf 3D}) + 22.4 + 21.3 \log_{10}(f_c) - 0.3(h_{\sf UT} - 1.5) \, \mathrm{dB}$, where $d_{\sf 3D} = \sqrt{h_{\sf BS}^2 + d^2}$ is the three-dimensional distance between the BS and user.
We assume that the BS and user heights are fixed at $h_{\sf BS} = 10\,\mathrm{m}$ and $h_{\sf UT} = 1.5\,\mathrm{m}$, respectively.
The system parameters include a carrier frequency $f_c \! = \! 10.5 \, \mathrm{GHz}$, a bandwidth $\mathrm{BW} \! = \! 300\, \mathrm{MHz}$, a noise figure $\sigma^2_{\sf NF} \! = \! 5\,\mathrm{dB}$, a shadowing variance $\sigma^2_{\sf SD} \! = \! 7.82\,\mathrm{dB}$, a noise power spectral density of $-174\,\mathrm{dBm/Hz}$, and an angular spread $\Delta_{\sf sd} \! = \! \pi/6$.


The considered benchmarks are summarized as follows:
\begin{itemize}
    \item {\bf ZF-DPC}: The ZF-DPC scheme proposed in \cite{caire2003TIT:DPC}, which provides an information-theoretic upper bound on the DL sum SE of MU-MIMO systems under perfect CSIT.
    \item {\bf ZF-DPC-WF}: The ZF-DPC scheme with a water-filling power allocation method proposed in \cite{caire2003TIT:DPC}.
    \item {\bf GPIP}: The conventional GPIP algorithm in \cite{choi2019TWC:GPI}.
    \item {\bf WMMSE}: The well-known WMMSE algorithm for sum SE maximization proposed in \cite{christensen2008TWC:WMMSE}.
    \item {\bf WMMSE-SAA}: The robust WMMSE-based algorithm utilizing sample average approximation, as proposed in \cite{kim2014book:WMMSE-SAA, joudeh2016TCOM:WMMSE-SAA}.
    \item {\bf R-WMMSE}: The R-WMMSE algorithm from \cite{zhao2023TSP:R-WMMSE}, which achieves low-complexity under perfect CSIT by avoiding high-dimensional matrix inversion.
    \item {\bf G-R-WMMSE (sequential)} and {\bf G-R-WMMSE (parallel)}: The G-R-WMMSE algorithms proposed in \cite{yoo2024WCL:G-R-WMMSE}, which extend the R-WMMSE concept to cell-free networks under imperfect CSIT using sequential and parallel update strategies, respectively. For comparisons, these algorithms are tailored to the considered single-cell system.
    \item {\bf RZF} and {\bf MRT}: Conventional linear RZF and MRT precoding schemes, serving as low-complexity baselines.
\end{itemize}

\subsection{Perfect CSIT Case}
\label{subsec:performance_Single-Cell_perfect}
In this subsection, we provide simulation results verifying the sum SE performance and computational complexity of the proposed S-GPIP method under perfect CSIT.

\begin{figure}[t]
    \centering
    \includegraphics[width=0.9\linewidth]{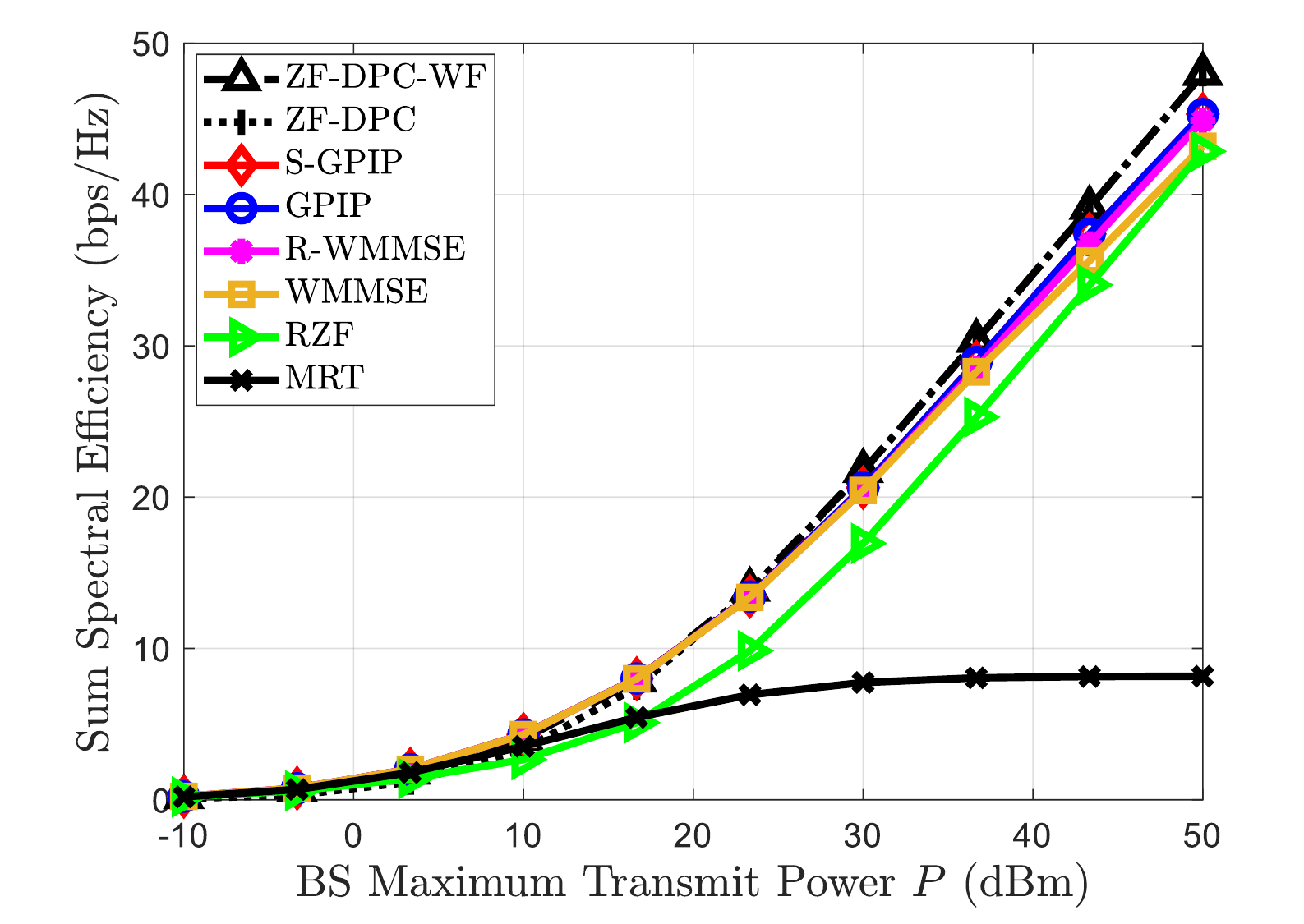}
    \caption{The sum SE versus the BS maximum transmit power budget $P$ for $N=32$ and $K=4$ under perfect CSIT.}
    \label{fig:SEvsPower_perfect}
\end{figure}
Fig.~\ref{fig:SEvsPower_perfect} illustrates the sum SE performance of various DL transmission strategies versus the BS maximum transmit power.
As shown in Fig.~\ref{fig:SEvsPower_perfect}, the proposed S-GPIP algorithm significantly outperforms conventional linear precoders (e.g., RZF and MRT) and closely tracks the performance of the conventional GPIP method.
In particular, the proposed method exhibits a slight performance gain over the WMMSE and R-WMMSE schemes in the high SNR regime $(P \! \geq \! 40\,\mathrm{dBm})$.
In the low SNR regime, the performance of S-GPIP approaches that of the ZF-DPC-WF benchmark.
However, this performance gap widens as the SNR increases, which is consistent with the findings reported in \cite{choi2019TWC:GPI}.

\begin{figure}[t]
    \centering
    \subfigure[]{
    \includegraphics[width=0.9\linewidth]{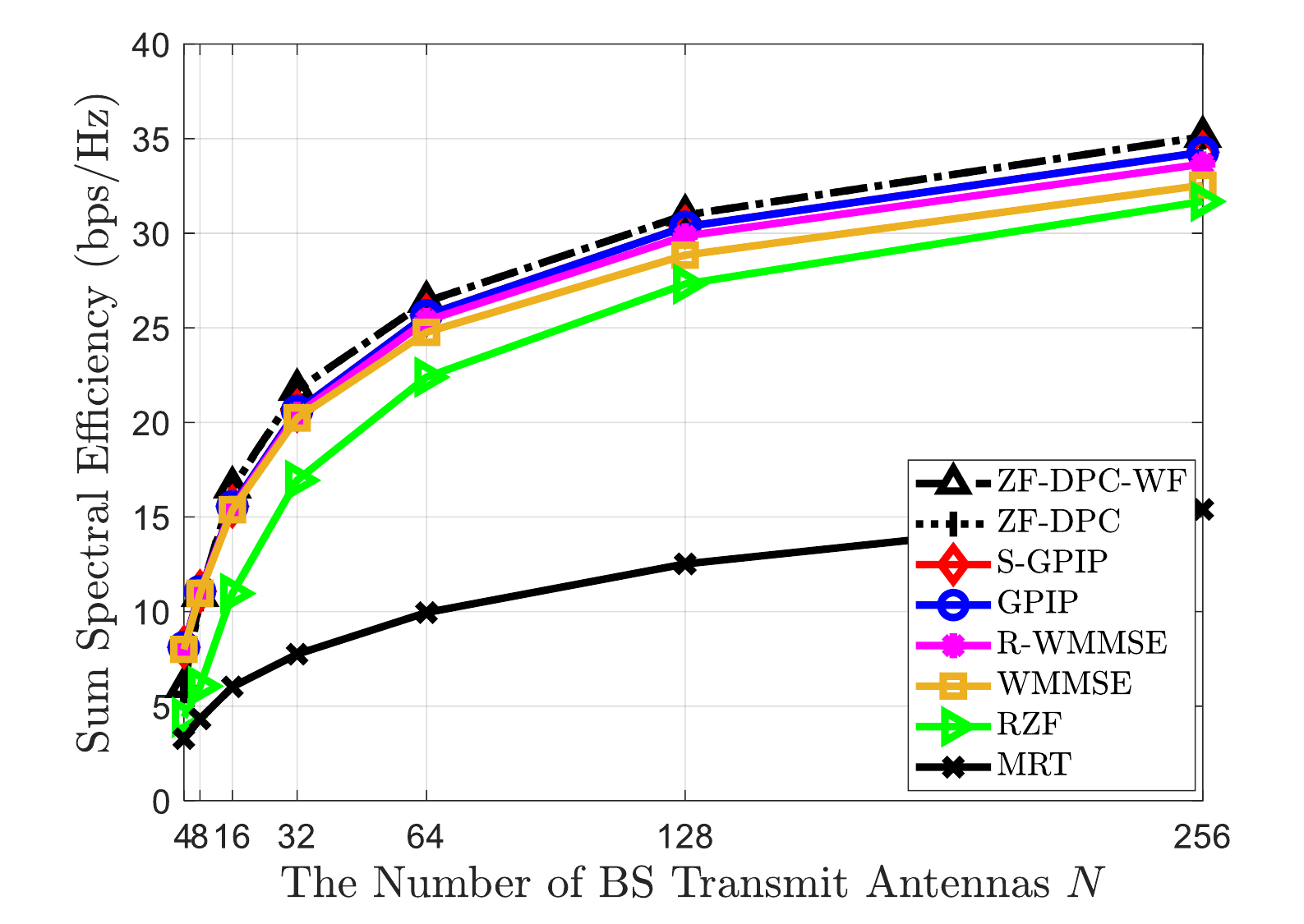}
    \label{fig:SEvsAntenna_perfect}
    }
    \subfigure[]{
    \includegraphics[width=0.9\linewidth]{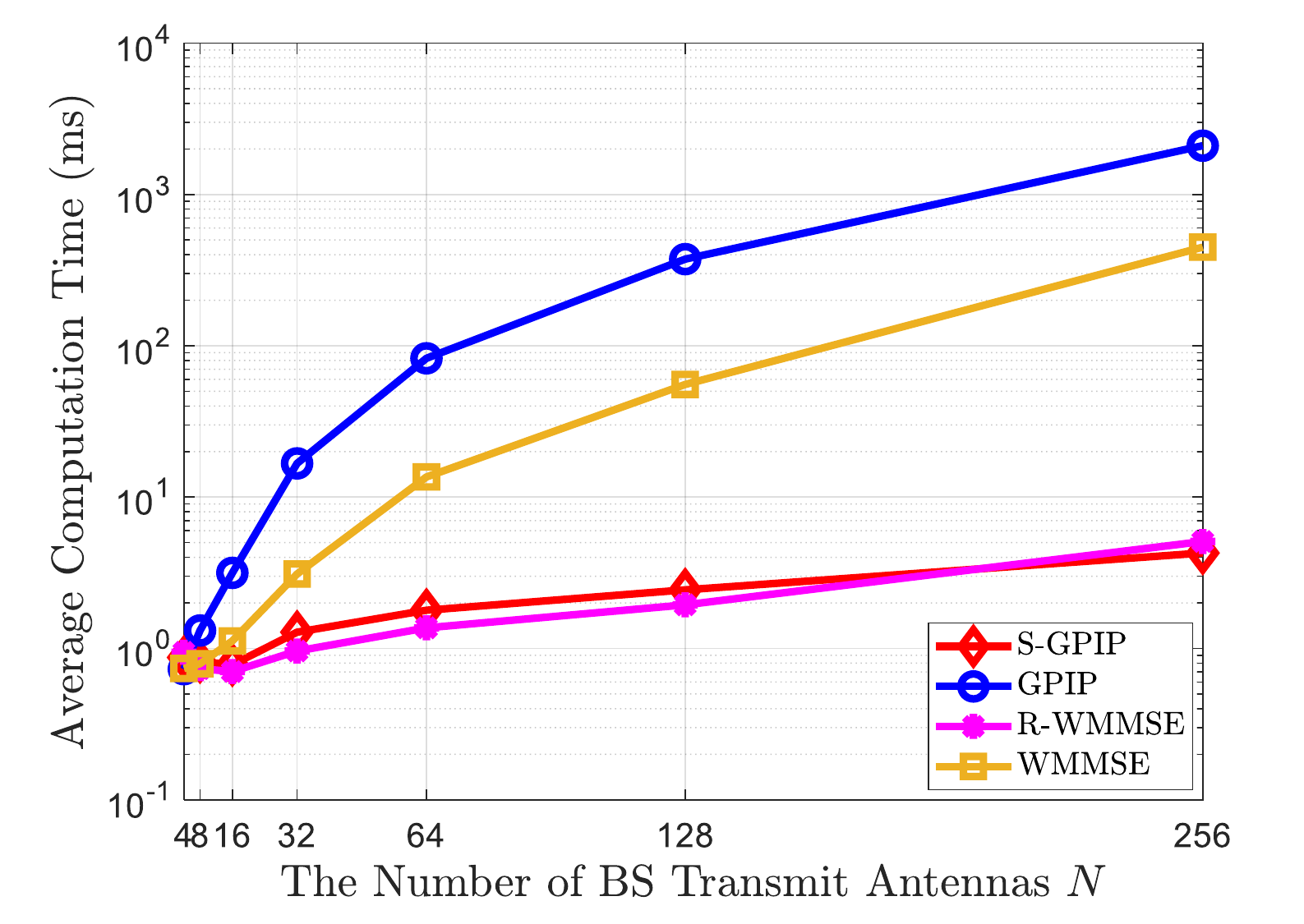}
    \label{fig:timevsAntenna_perfect}
    }
    \caption{(a) The sum SE and (b) computation time versus the BS transmit antennas $N$ for $K=4$ and $P=30\,\mathrm{dBm}$ under perfect CSIT.}
    \label{fig:antennas}
\end{figure}
Fig.~\ref{fig:antennas} shows the sum SE and average computation time versus the number of BS transmit antennas.
As shown in Fig.~\ref{fig:SEvsAntenna_perfect}, the proposed S-GPIP algorithm achieves the highest sum SE across the entire range of antennas numbers among all benchmarks, with the exception of the ideal ZF-DPC schemes.
Specifically, Fig.~\ref{fig:timevsAntenna_perfect} highlights the primary advantage of S-GPIP: its computation time exhibits a marginal increase with the number of antennas, whereas conventional GPIP and WMMSE scale rapidly with $N$.
This results in S-GPIP being up to 100 times faster at $N=256$, as its dominant complexity depends on the number of users $K$ rather than $N$.
Furthermore, while the S-GPIP algorithm demonstrates computation time comparable to the R-WMMSE algorithm around $N=200$, it surpasses the R-WMMSE in both sum SE and computational efficiency for larger arrays (e.g., $N=256$).
These results confirm that the proposed algorithm achieves sum SE performance comparable to state-of-the-art methods with a substantially reduced computational burden.

\subsection{Imperfect CSIT Case}
\label{subsec:performance_Single-Cell_imperfect}
In this subsection, we evaluate the robustness of the proposed method against channel estimation errors by assessing their achievable sum SE and computation time under the imperfect CSIT.
In the following figures, dashed curves denote benchmarks that do not utilize error covariance information.
This convention is maintained throughout the subsequent simulations unless specified otherwise.

\begin{figure}[t]
    \centering
    \includegraphics[width=0.9\linewidth]{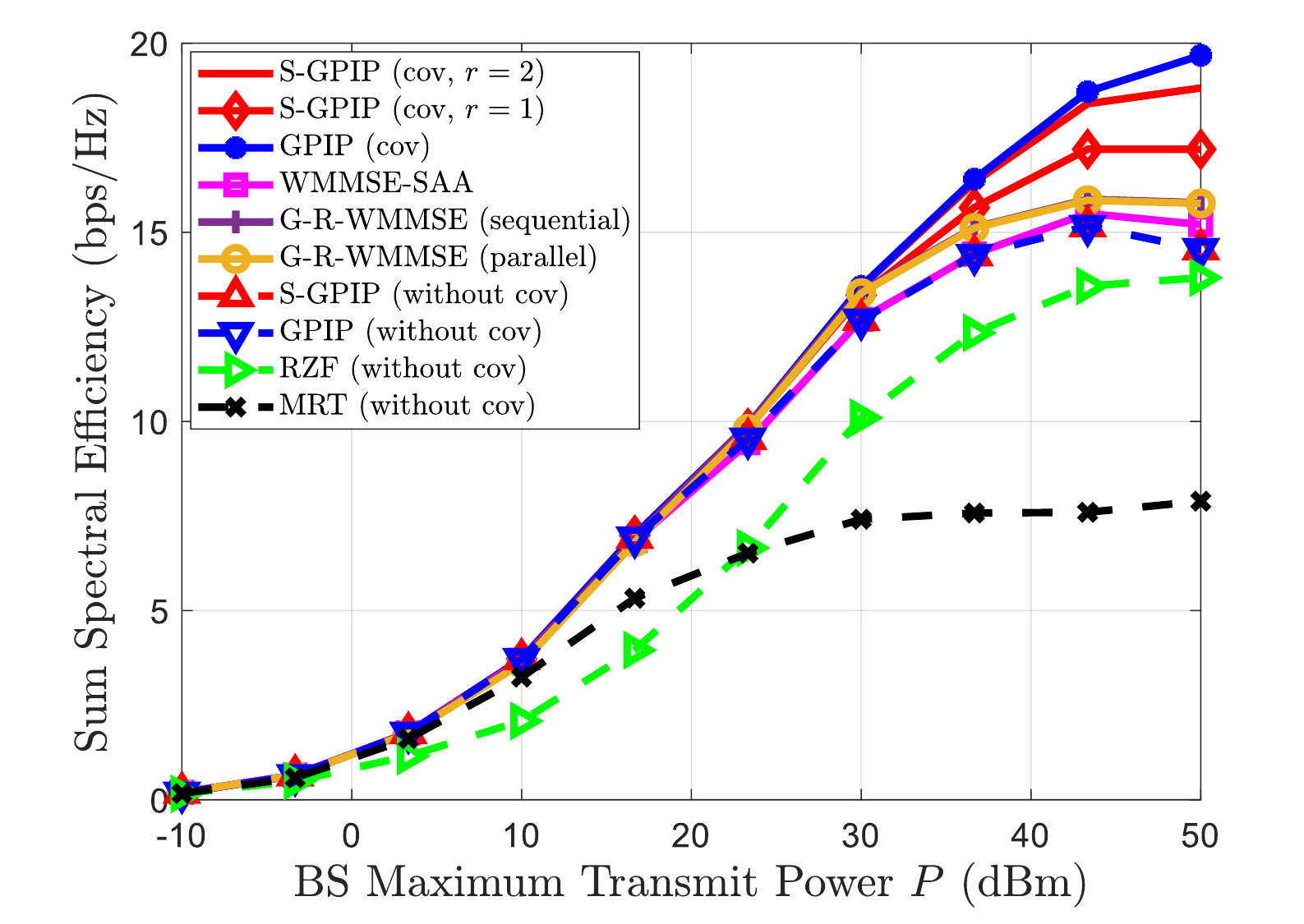}
    \caption{The sum SE versus the BS maximum transmit power budget $P$ for $N=32$, $K=4$, and $\kappa=0.3$ under imperfect CSIT.}
    \label{fig:SEvsPower_imperfect}
\end{figure}
Fig.~\ref{fig:SEvsPower_imperfect} illustrates the sum SE performance versus the BS maximum transmit power.
Specifically, we examine the proposed S-GPIP (cov) algorithms with rank parameters $r=1$ and $r=2$. 
To construct the precoder defined in \eqref{eq:reformulated_opt_precoding_vector_imperfect}, these variants correspond to utilizing the principal and the two dominant eigenpairs in \eqref{eq:power_iteration_imperfect}, respectively.
As depicted in Fig.~\ref{fig:SEvsPower_imperfect}, the proposed S-GPIP (cov) schemes outperform all existing benchmarks, strictly surpassed only by the conventional GPIP with error covariance (GPIP (cov)). 
Notably, the proposed S-GPIP (cov, $r=2$) algorithm achieves SE performance that nearly matches that of the GPIP (cov), even in the high SNR regime.
We recall that GPIP (cov) utilizes the full error covariance matrix without reducing its computational complexity, i.e., the complexity remains $\mathcal{O}(N^3K)$.
This demonstrates that the proposed algorithm attains sum SE performance comparable to the state-of-the-art GPIP (cov) under imperfect CSIT, while offering significantly reduced computational complexity.

\begin{figure}[t]
    \centering
    \subfigure[]{
    \includegraphics[width=0.9\linewidth]{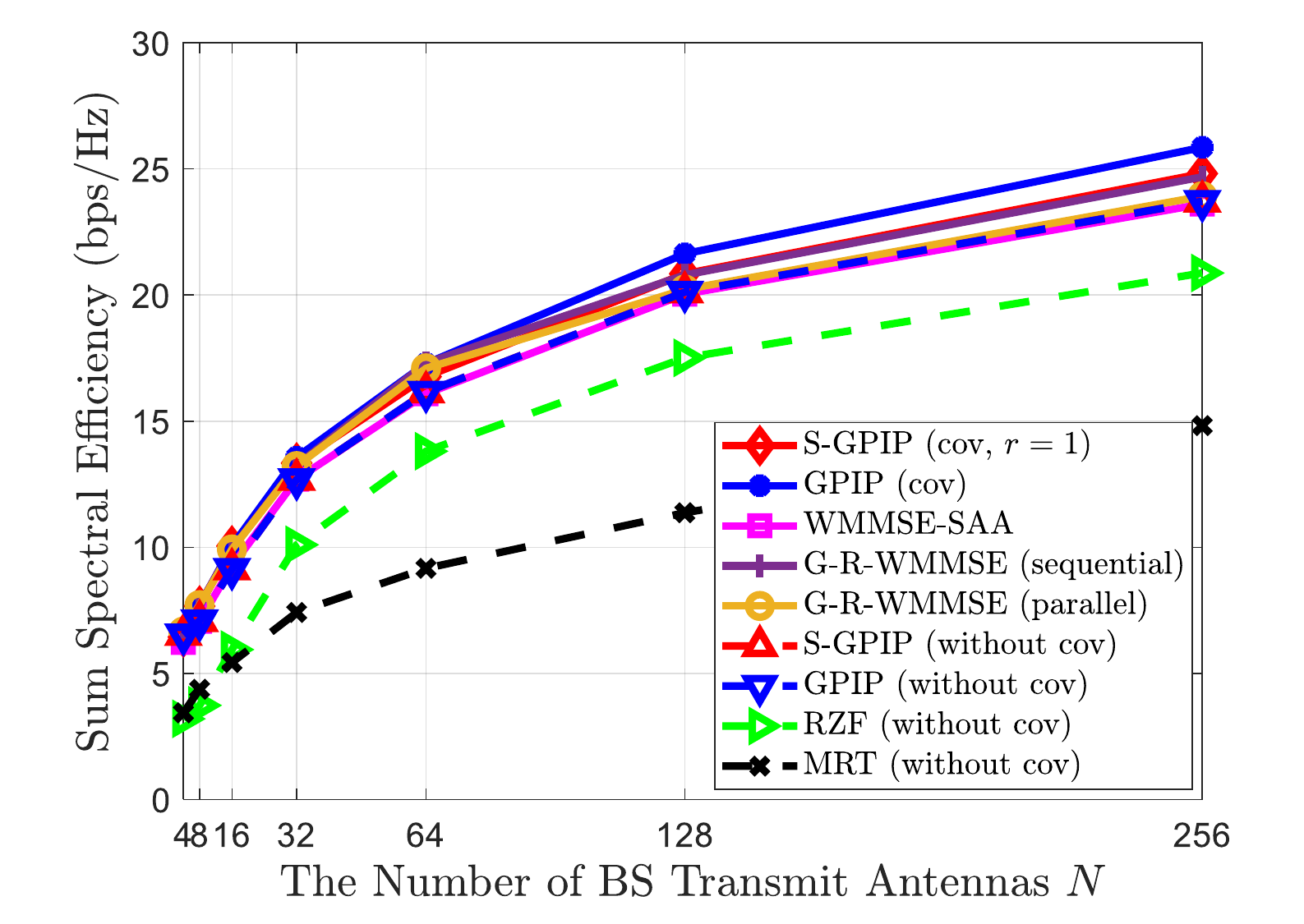}
    \label{fig:SEvsAntenna_imperfect}
    }
    \subfigure[]{
    \includegraphics[width=0.9\linewidth]{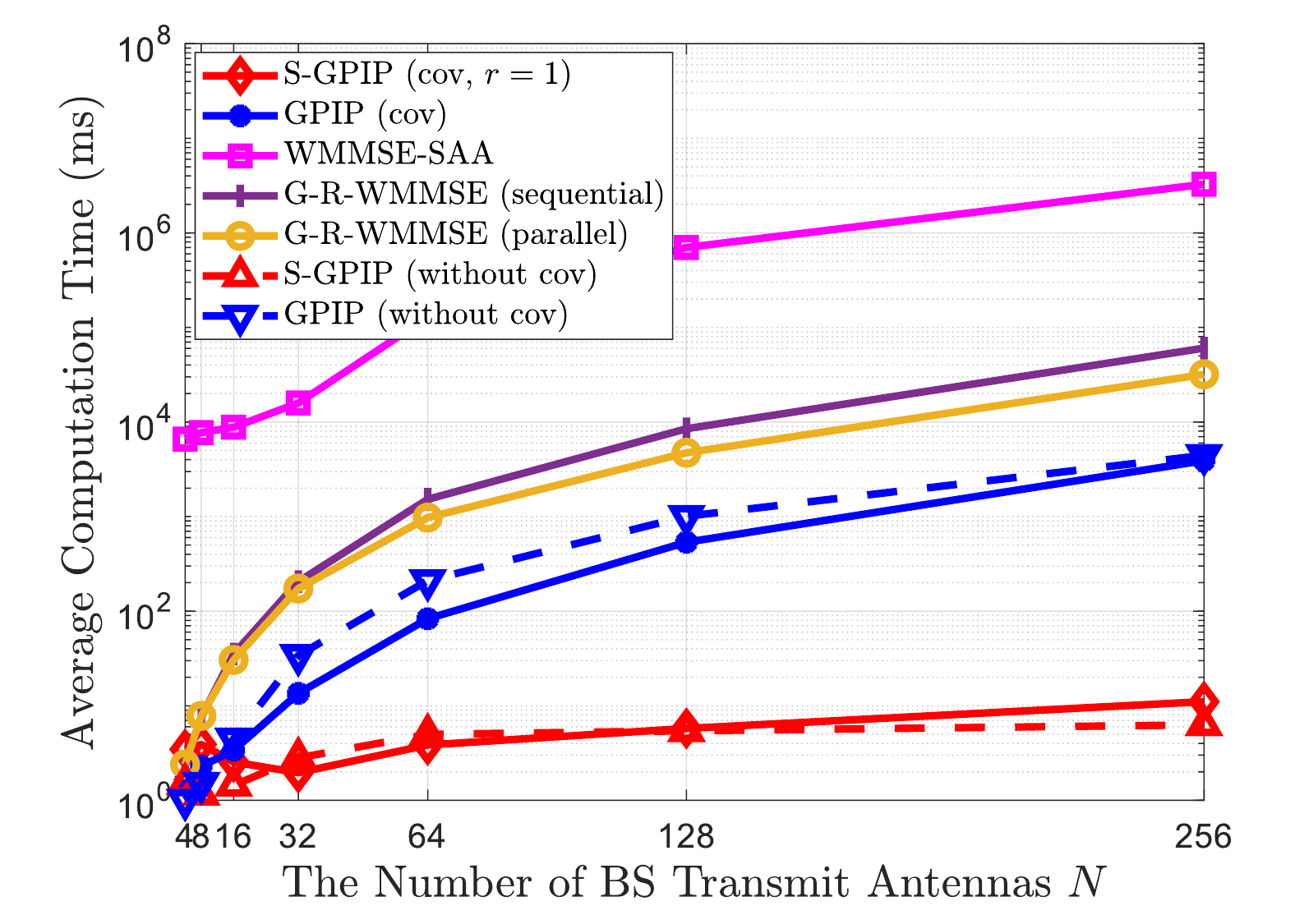}
    \label{fig:timevsAntenna_imperfect}
    }
    \caption{(a) The sum SE and (b) computation time versus the BS transmit antennas $N$ for $K=4$, $P=30\,\mathrm{dBm}$, and $\kappa=0.3$ under imperfect CSIT.}
    \label{fig:antennas_imp}
\end{figure}
Fig.~\ref{fig:antennas_imp} shows the SE performance and average computation time versus the number of BS transmit antennas.
As shown in Fig.~\ref{fig:SEvsAntenna_imperfect}, the proposed S-GPIP (cov) algorithm achieves sum SE performance that closely matches that of the GPIP (cov) method, while outperforming all other benchmarks.
Specifically, this high performance is maintained even by the S-GPIP (cov, $r=1$) variant, which utilizes only the principal eigenvector in \eqref{eq:power_iteration_imperfect}.
This demonstrates the algorithm's robustness against channel estimation uncertainty, consistent with the observations in Fig.~\ref{fig:SEvsPower_imperfect}.
Regarding the average computation time shown in Fig.~\ref{fig:timevsAntenna_imperfect}, while the computational load of the all benchmarks scales rapidly with $N$, the complexity of the proposed S-GPIP (cov) algorithm remains nearly constant, exhibiting only marginal growth even for large-scale antenna arrays.
Consequently, the proposed scheme achieves an exceptional balance between high sum SE performance and low computational complexity, highlighting its scalability and practicality for massive MIMO deployments.

\begin{figure}[t]
    \centering
    \subfigure[]{
    \includegraphics[width=0.9\linewidth]{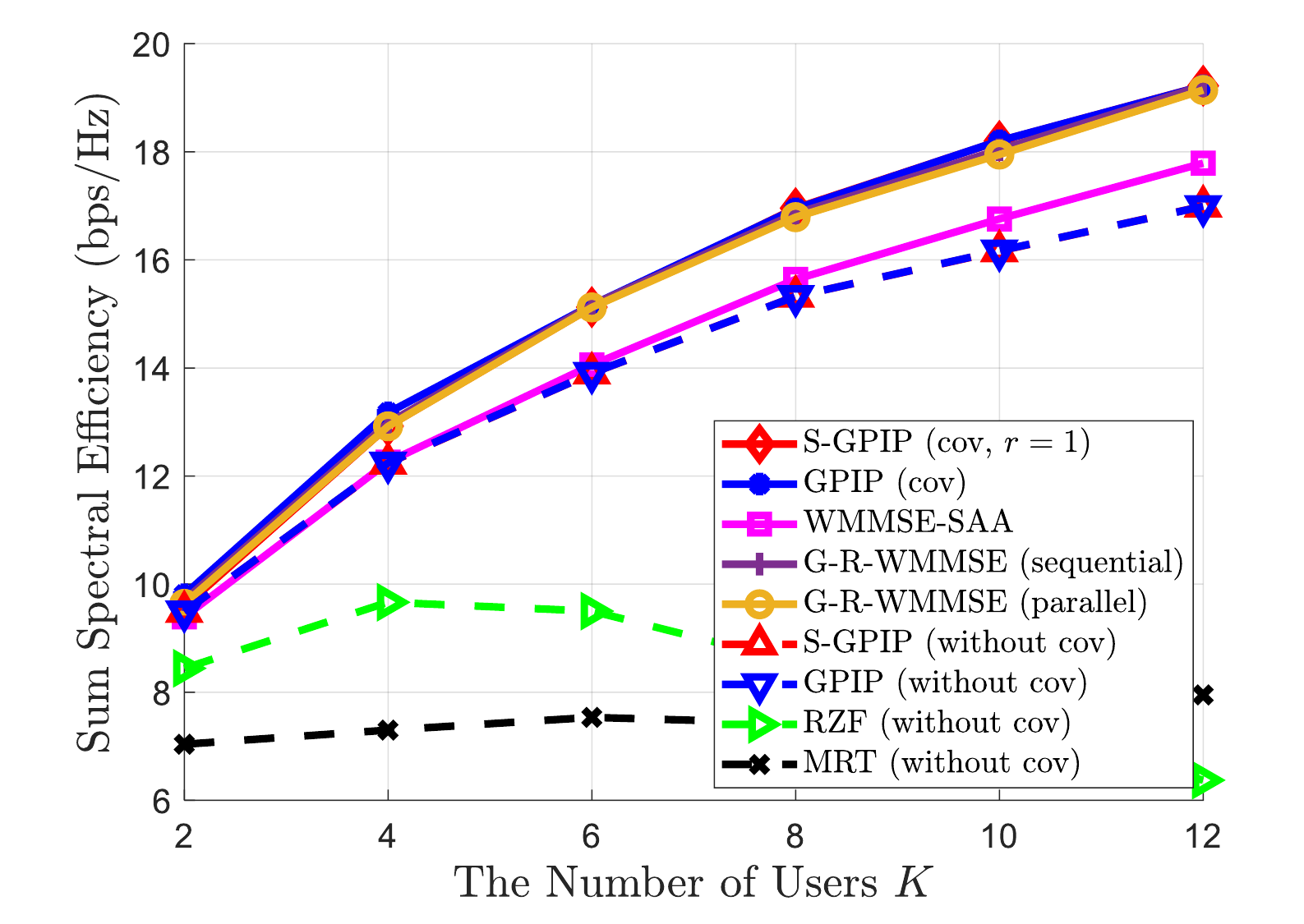}
    \label{fig:SEvsUser_imperfect}
    }
    \subfigure[]{
    \includegraphics[width=0.9\linewidth]{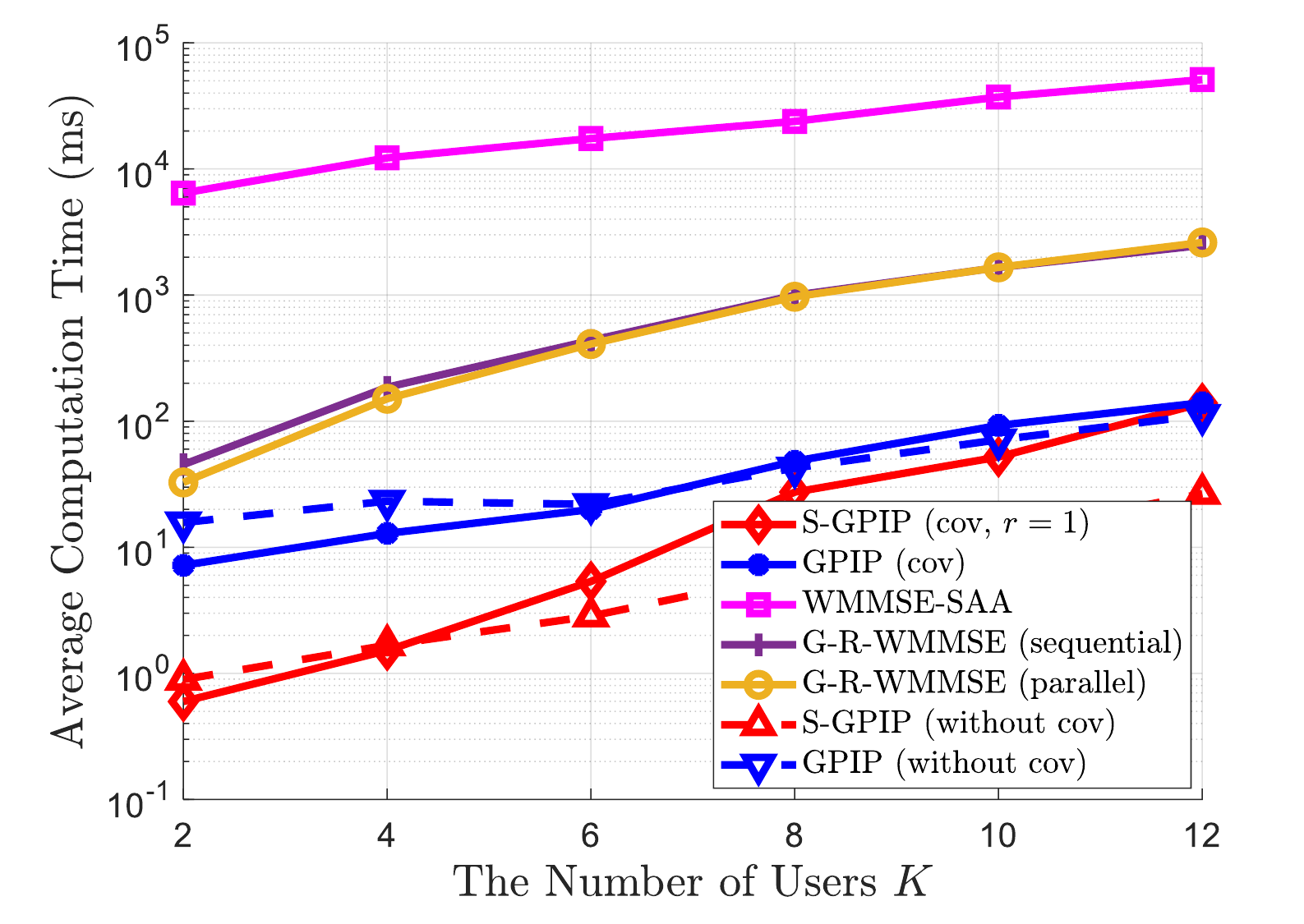}
    \label{fig:timevsUser_imperfect}
    }
    \caption{(a) The sum SE and (b) computation time versus DL users $K$ for $N=32$, $P=30\,\mathrm{dBm}$, and $\kappa=0.3$ under imperfect CSIT.}
    \label{fig:users}
\end{figure}
Fig.~\ref{fig:users} presents the sum SE and average computation time versus the number of DL users $K$. 
As depicted in Fig.~\ref{fig:SEvsUser_imperfect}, the proposed S-GPIP (cov) algorithm closely tracks the performance of the GPIP (cov) across the entire range of user counts. 
However, Fig.~\ref{fig:timevsUser_imperfect} reveals that the computation time of the proposed S-GPIP (cov) algorithm scales significantly with $K$. 
This observation aligns with the theoretical complexity analysis in Table~\ref{Table1}, as the dominant matrix inversion $\bar{\bB}_{\sf KKT}^{-1} ( \bar{\bv}^{(t-1)} )$ in Algorithm~\ref{alg:two} imposes a complexity of $\CMcal{O}(K^4)$. 
Nevertheless, in typical massive MIMO scenarios where $N \gg K$, the proposed scheme maintains a substantial computational advantage over the baselines, whose complexity is dominated by $N$.

\subsection{Convergence}

\begin{figure}[t]
    \centering
    \subfigure[]{
    \includegraphics[width=0.9\linewidth]{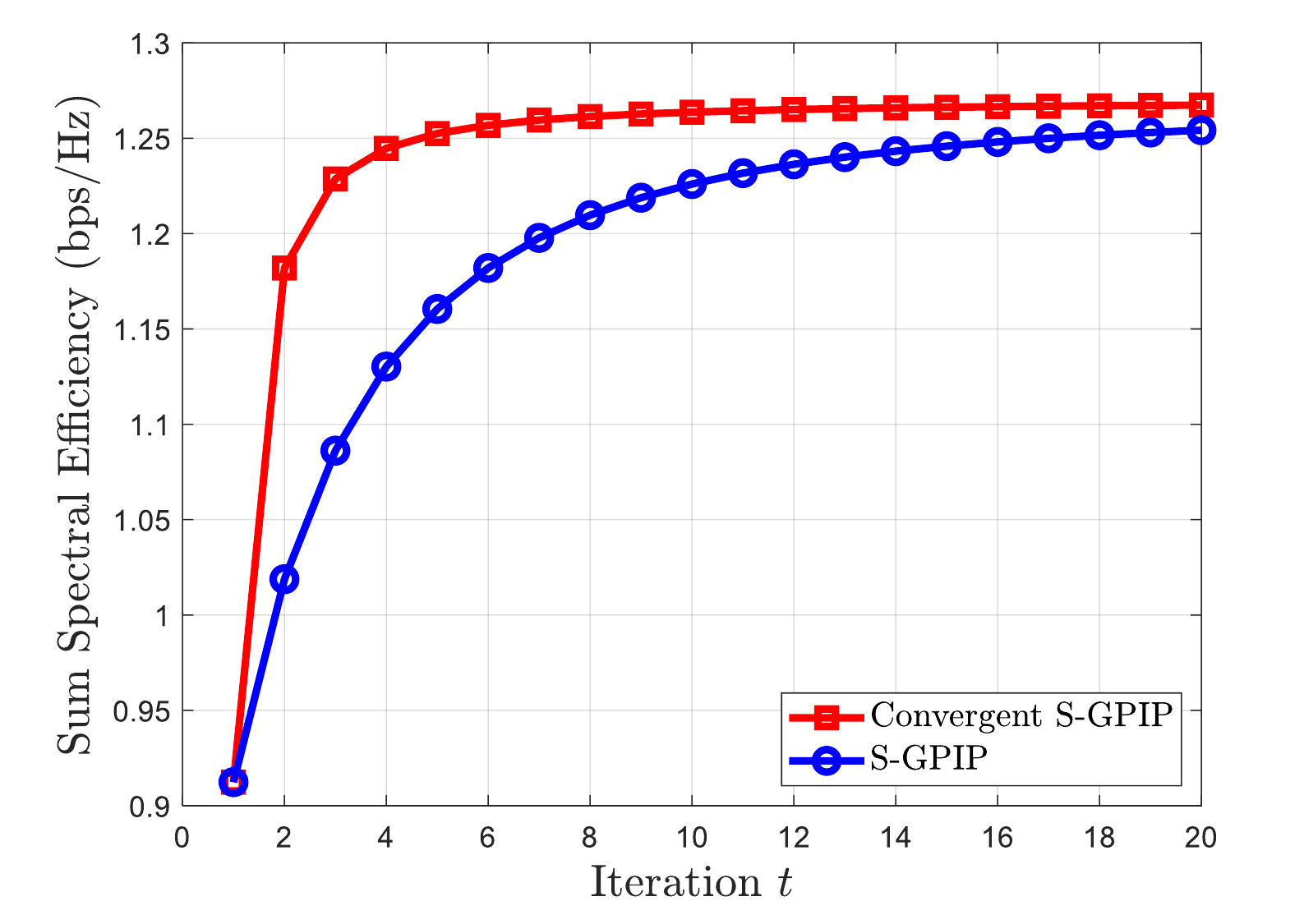}
    \label{fig:conv_P0dBm_perfect}
    }
    \subfigure[]{
    \includegraphics[width=0.9\linewidth]{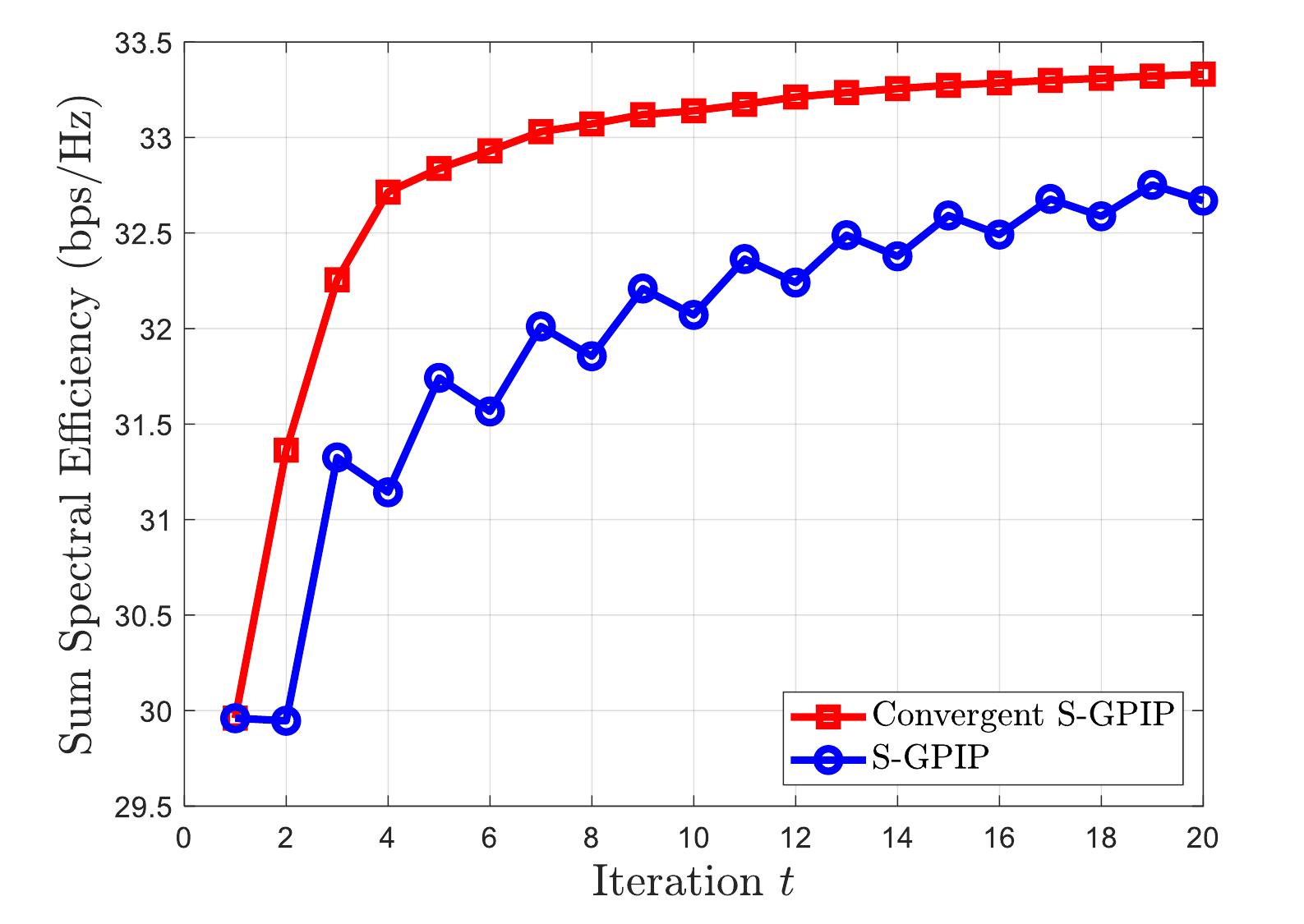}
    \label{fig:conv_P40dBm_perfect}
    }
    \caption{The convergence over iterations $t$ for (a) $P=0\,\mathrm{dBm}$ and (b) $P=40\,\mathrm{dBm}$, with $N=32$ and $K=4$ under perfect CSIT.}
    \label{fig:convergence}
\end{figure}

Fig.~\ref{fig:convergence} illustrates the sum SE convergence of the proposed S-GPIP algorithm and convergent GPIP algorithm over iterations $t$ for perfect CSIT.
We compare the S-GPIP ($\eta=1/2$) with the convergent S-GPIP algorithm which dynamically adjusts $\eta$ via backtracking line search to satisfy the sufficient ascent condition in \eqref{eq:eta_condition}.
As shown in Fig.~\ref{fig:conv_P0dBm_perfect}, in the low SNR regime, both methods reveal stable convergence to a local stationary point, corroborating the theoretical analysis in Section~\ref{sec:conv_pg}.
In particular, it is shown that the convergent S-GPIP algorithm converges faster in the considered system.
This improvement stems from the dynamic adaptation of the step size $\eta$, which allows the algorithm to take fine steps along the preconditioned gradient direction.

Fig.~\ref{fig:conv_P40dBm_perfect} depicts the convergence in the high SNR regime for perfect CSIT.
Unlike the low SNR regime, the S-GPIP exhibits oscillatory convergence.
This is because, in this high SNR regime, the objective function becomes highly non-linear and steep, resulting in a large Lipschitz constant $L_{g}$.
Consequently, the fixed $\eta = 1/2$ frequently violates the theoretical upper bound in \eqref{eq:eta_condition}, causing the update to overshoot.
In contrast, the convergent S-GPIP adaptively shrinks $\eta$ to enforce the sufficient ascent condition, ensuring a smooth, monotonic ascent and even faster convergence to a superior stationary point.
Overall, it is demonstrated that the convergent GPIP can guarantee the convergence with faster convergence.

\section{Conclusion}
\label{sec:conclusion}
In this paper, we developed a scalable and computationally efficient precoding framework for massive MIMO systems under both perfect and imperfect CSIT. 
By leveraging the low-dimensional subspace property of optimal precoders, we reformulated the high-dimensional beamforming problem into a lower-dimensional weight optimization whose dimension scales with the number of users rather than antennas. 
This reformulation enables the GPIP framework to overcome the conventional cubic complexity bottleneck associated with large antenna arrays.
We further extended this subspace-based principle to the imperfect CSIT setting by showing that stationary solutions lie in a combined subspace spanned by the estimated channel and error covariance matrices. 
This result provides a structured approach to robust precoding and motivates a low-rank approximation that preserves scalability while maintaining performance.
In addition, by exploiting the rank-one update structure via the Sherman–Morrison formula, we further reduced the matrix inversion complexity. 
Interpreting the GPIP update as a PPGA method, we also established convergence guarantees and developed a provably convergent scalable algorithm.
Numerical results demonstrated that the proposed methods achieve the highest SE performance compared to state-of-the-art linear precoders while substantially reducing computational complexity and ensuring stable convergence in the large antenna regime, highlighting their practicality for large-scale massive MIMO systems.
\appendices
\section{Proof of Lemma \ref{Lemma1}}
\label{app:Lemma1}

We first denote the objective function as 
\begin{align}
    L(\bar{\bw}) = \mathrm{log}_2 \prod_{k=1}^{K} \left(\frac{\bar{\bw}^{\sf H} \bA_{k} \bar{\bw}}{\bar{\bw}^{\sf H} \bB_{k} \bar{\bw}} \right).\label{eq:Lagrangian}
\end{align}
We need to compute the derivative of \eqref{eq:Lagrangian} to obtain the first-order optimality condition.
Let us define $\lambda(\bar{\bw})$ as
\begin{align}
    \lambda(\bar{\bw}) = \prod_{k=1}^{K} \left( \frac{\bar{\bw}^{\sf H}\bA_{k}\bar{\bw}}{\bar{\bw}^{\sf H} \bB_{k} \bar{\bw}}\right). 
\end{align}
Thus, we have $L(\bar{\bw})= \mathrm{log_2} \, \lambda(\bar{\bw})$.
Subsequently, we compute a derivative of $L(\bar{\bw})$ as
\begin{align}
    \frac{\partial L(\bar{\bw})}{\partial \bar{\bw}^{\sf H}} = \frac{1}{\lambda(\bar{\bw})\mathrm{ln}\,2} \frac{\partial \lambda(\bar{\bw})}{ \partial \bar{\bw}^{\sf H}}. \label{eq:derivative_Lagrangian}
\end{align}
Hence, the derivative of $\lambda(\bar {\bw})$ is computed as
\begin{align}
    \label{eq:dlambda}
    \frac{\partial \lambda(\bar{\bw})}{ \partial \bar{\bw}^{\sf H}} = 2 \lambda(\bar{\bw}) \left(\sum_{k=1}^{K} \left( \frac{\bA_{k}\bar{\bw}}{\bar{\bw}^{\sf H}\bA_{k}\bar{\bw}} - \frac{\bB_{k}\bar{\bw}}{\bar{\bw}^{\sf H}\bB_{k}\bar{\bw}} \right)\right).
\end{align}
Setting \eqref{eq:dlambda} equal to zero with reorganization provides the stationary condition as
\begin{align}
    \label{eq:KKT_cond}
    \lambda_{\sf num}(\bar{\bw})\left(\sum_{k=1}^{K} \frac{\bA_{k}}{\bar{\bw}^{\sf H}\bA_{k}\bar{\bw}}\right)\bar{\bw} = \lambda(\bar{\bw})\lambda_{\sf den}(\bar{\bw})\left(\sum_{k=1}^{K}\frac{\bB_{k}}{\bar{\bw}^{\sf H}\bB_{k}\bar{\bw}}\right)\bar{\bw}.
\end{align}
Accordingly, we derive the necessary condition of the first-order optimality condition from \eqref{eq:KKT_cond} as
\begin{align}
    \bA_{\sf KKT}(\bar{\bw})\bar{\bw} = \lambda(\bar{\bw})\bB_{\sf KKT}(\bar{\bw})\bar{\bw}, \label{eq:necessary_condition}
\end{align}
where $\bA_{\sf KKT}(\bar{\bw})$ and $\bB_{\sf KKT}(\bar{\bw})$ are defined in \eqref{eq:A_KKT} and \eqref{eq:B_KKT}, respectively. 
Since $\bB_{\sf KKT}(\bar{\bw})$ is Hermitian and $N>K$, it is expected to be non-singular. 
This completes the proof.
\qed

\section{Proof of Proposition~\ref{proposition2}}
\label{app:impCSIT}
We prove Proposition \ref{proposition2} by analyzing the KKT conditions of the problem \eqref{avg_problem}.
Specifically, we establish that the Lagrange multiplier $\lambda$ must be strictly positive by using a proof by contradiction.
For any nontrivial stationary point $\bff_{k}$ and Lagrange multiplier $\lambda$, the following KKT conditions hold:
\begin{subequations}
    \label{app:KKTconditions}
    \begin{align}
        &\frac{\partial}{\partial\bff^{\sf H}_{k}} \left[\sum_{k=1}^{K} \bar{R}^{\sf lb}_{k}(\bF) - \lambda (\mathrm{tr}(\bF \bF^{\sf H})-1)\right] = 0, \label{stationarity}\\
        &\lambda(\mathrm{tr}(\bF \bF^{\sf H})-1) = 0,\\
        &\mathrm{tr}(\bF \bF^{\sf H}) \leq 1,\\
        &\lambda \geq 0.
    \end{align}
\end{subequations}
We calculate the stationarity condition in \eqref{stationarity} as the derivative of $\sum_{k=1}^{K} \bar{R}^{\sf lb}_{k}(\bF)$ with respect to $\bff^{\sf H}_{k}$. To this end, we first obtain the derivative of $\bar{R}^{\sf lb}_{k}(\bF)$ with respect to $\bff^{\sf H}_{k}$ as
\begin{align}
    \frac{\partial}{\partial\bff^{\sf H}_{k}} \bar{R}^{\sf lb}_{k}(\bF) =A^{-1}_{k}(\bF)(\hat{\bh}_{k} \hat{\bh}^{\sf H}_{k} \bff_{k} + \boldsymbol{\Phi}_{k}\bff_{k}) - B^{-1}_{k}(\bF)\boldsymbol{\Phi}_{k}\bff_{k},
\end{align}
where
$A^{-1}_{k}(\bF) = \left( \sum_{i=1}^{K}|\hat{\bh}^{\sf H}_{k} \bff_{i}|^2 + \sum_{i=1}^{K}\bff_{i}^{\sf H} \boldsymbol{\Phi}_{k}\bff_{i} + \frac{\sigma^{2}}{P} \right)^{-1}$ and $B^{-1}_{k}(\bF) = \left( \sum_{i=1,i\ne k}^{K}|\hat{\bh}^{\sf H}_{k} \bff_{i}|^2 + \sum_{i=1}^{K}\bff_{i}^{\sf H} \boldsymbol{\Phi}_{k}\bff_{i} + \frac{\sigma^{2}}{P} \right)^{-1}$.
Note that the constant factor $1/\ln 2$ arising from the derivative of $\log_2(\cdot)$ is absorbed into the Lagrange multiplier $\lambda$ for simplicity.
We then calculate the derivative of $\sum_{j=1,j \ne k}^{K}\bar{R}^{\sf lb}_{j}(\bF)$ with respect to $\bff^{\sf H}_{k}$, i.e.,
\begin{align}
    \sum_{j=1,j \ne k}^{K}\frac{\partial}{\partial\bff^{\sf H}_{k}} \bar{R}^{\sf lb}_{j}(\bF) =& \! \sum_{j=1,j \ne k}^{K} \! A^{-1}_{j}(\bF) (\hat{\bh}_{j} \hat{\bh}^{\sf H}_{j} \bff_{k} + \boldsymbol{\Phi}_{j}\bff_{k}) \nonumber \\
    &- \! \sum_{j=1,j \ne k}^{K} \! B^{-1}_{j}(\bF) (\hat{\bh}_{j} \hat{\bh}^{\sf H}_{j} \bff_{k} + \boldsymbol{\Phi}_{j}\bff_{k}).
\end{align}
To proceed with the contradiction, we assume the contrary that $\lambda \! = \! 0$.
Consequently, the following stationarity condition holds:
\begin{align}
    \label{app:stationarity}
    \frac{\partial}{\partial\bff^{\sf H}_{k}} \bar{R}^{\sf lb}_{k}(\bF) + \sum_{j=1,j\ne k}^{K} \frac{\partial}{\partial\bff^{\sf H}_{k}} \bar{R}^{\sf lb}_{j}(\bF) = 0.
\end{align}
By multiplying $\bff_{k}^{\sf H}$ on the left of the condition in \eqref{app:stationarity}, and summing over $k=1,...,K$ and rearranging the terms, we represent the following equation as
\begin{align}
    \label{app:proof_imp}
    &\sum_{k=1}^{K} \left( A^{-1}_{k}(\bF) \! \cdot \! \sum_{j=1}^{K} \left( |\hat{\bh}_{k}^{\sf H} \bff_{j}|^{2} + \bff_{j}^{\sf H}\boldsymbol{\Phi}_{k}\bff_{j} \right) \! \right) = \nonumber\\
    &\sum_{k=1}^{K} \left( B^{-1}_{k}(\bF) \! \cdot \! \left( \sum_{j=1, j\ne k}^{K}|\hat{\bh}_{k}^{\sf H} \bff_{j}|^{2} + \sum_{j=1}^{K}\bff_{j}^{\sf H}\boldsymbol{\Phi}_{k}\bff_{j} \right) \right).
\end{align}
Using $v(v+c)^{-1} = 1 - c(v+c)^{-1}$, where $v$ and $c$ are a scalar variable and a constant value, respectively, we can represent the equation in \eqref{app:proof_imp} as
\begin{align}
    \label{app:proof_imp2}
    &\sum_{k=1}^{K} \left(\sum_{i=1}^{K}|\hat{\bh}^{\sf H}_{k} \bff_{i}|^2 + \sum_{i=1}^{K}\bff_{i}^{\sf H} \boldsymbol{\Phi}_{k}\bff_{i} + \frac{\sigma^{2}}{P}\right)^{-1} = \nonumber \\
    &\sum_{k=1}^{K} \left(\sum_{i=1, i\ne k}^{K}|\hat{\bh}^{\sf H}_{k} \bff_{i}|^2 + \sum_{i=1}^{K}\bff_{i}^{\sf H} \boldsymbol{\Phi}_{k}\bff_{i} + \frac{\sigma^{2}}{P} \right)^{-1}.
\end{align}
Note that the equation in \eqref{app:proof_imp2} holds if and only if $\hat{\bh}^{\sf H}_{k} \bff_{k} \! = \! 0$ for all $k$.
This condition implies zero desired signal power for all users, yielding a zero sum SE.
This trivial solution contradicts the premise that $\bff_{k}$ is a nontrivial stationary point. 
Thus, the optimal Lagrange multiplier must satisfy $\lambda^{\star} > 0$.
Consequently, we define the optimal precoding vector $\bff^{\star}_{k}$ as
\begin{align}
    \bff^{\star}_{k} =& \frac{1}{\lambda^{\star}} \Bigg( \sum_{j=1}^{K} A^{-1}_{j}(\bF) \hat{\bh}_{j} \hat{\bh}^{\sf H}_{j} \bff_{k} - \sum_{j=1,j \ne k}^{K} B^{-1}_{j}(\bF) \hat{\bh}_{j} \hat{\bh}^{\sf H}_{j} \bff_{k} = \nonumber \\
    &+ \sum_{j=1}^{K} \left( A^{-1}_{j}(\bF) - B^{-1}_{j}(\bF) \right) \boldsymbol{\Phi}_{j}\bff_{k}\Bigg).
\end{align}
It implies that any nontrivial stationary point $\bff_{k}$ of the problem in \eqref{avg_problem} must lie in the range space of the estimated channel $\hat{\bH}$ and all error covariance matrices for users ${{\boldsymbol{\Phi}}} = \left[{\boldsymbol{\Phi}}_{1},...,{\boldsymbol{\Phi}}_{K} \right]$, and the corresponding Lagrange multiplier $\lambda$ is positive.
This completes the proof.
\qed

\bibliographystyle{IEEEtran}
\bibliography{bibtex}
\end{document}